\documentclass{article}

\usepackage{arxiv}

\usepackage[utf8]{inputenc} 
\usepackage[T1]{fontenc}    
\usepackage{hyperref}       
\usepackage{url}            
\usepackage{booktabs}       
\usepackage{amsfonts}       
\usepackage{nicefrac}       
\usepackage{microtype}      
\usepackage{lipsum}
\usepackage{graphicx}
\usepackage{amssymb}
\usepackage{amsmath} 
\usepackage{amsthm}

\newtheorem*{theorem*}{Theorem}

\newtheorem{lemma}{Lemma}
\newtheorem*{lemma*}{Lemma}

\usepackage{algorithm}
\usepackage{algorithmicx,algpseudocode}
\usepackage{caption}
\usepackage{subcaption}
\usepackage{thmtools, thm-restate}
\usepackage{tikz}
\newcommand*\circled[1]{\tikz[baseline=(char.base)]{
		\node[shape=circle,draw,inner sep=2pt] (char) {#1};}}
\usepackage{color,enumitem}
\title{Distributed Q-Learning with State Tracking for Multi-agent Networked Control}

\author{
 Hang Wang \\
  Arizona State University\\
  Tempe, AZ 85287-7206 \\
  \texttt{hwang442@asu.edu} \\
   \And
 Sen Lin \\
 Arizona State University\\
 Tempe, AZ 85287-7206 \\
 \texttt{slin70@asu.edu} \\
  \And
   Hamid Jafarkhani \\
  University of California, Irvine\\
  Irvine, CA 92697-2625 \\
  \texttt{hamidj@uci.edu} \\
    \And
  Junshan Zhang \\
  Arizona State University\\
  Tempe, AZ 85287-7206 \\
  \texttt{Junshan.Zhang@asu.edu} \\
}

\begin{document}
\maketitle
\begin{abstract}
This paper studies distributed Q-learning for Linear Quadratic Regulator (LQR) in a multi-agent network. The existing results often assume that agents can observe the global system state, which may be  infeasible in large-scale systems due to privacy concerns or communication constraints. In this work, we consider a setting  with unknown system models and no centralized coordinator. We devise a state tracking (ST) based Q-learning algorithm to design optimal controllers for agents. Specifically, we assume that agents maintain local estimates of the global state based on their local information and communications with neighbors. At each step, every agent updates its local global state estimation, based on which it solves an approximate Q-factor locally through policy iteration. Assuming decaying injected excitation noise during the policy evaluation, we prove that the local estimation  converges to the true global state, and  establish the convergence  of the proposed distributed ST-based Q-learning algorithm. The experimental studies corroborate our theoretical results by showing that our proposed method achieves comparable performance with the  centralized  case.
\end{abstract}


\section{Introduction}
Distributed control of multi-agent systems (MASs) has garnered much attention in the past decade, due to its wide applicability in real-world problems, e.g., firefighting unmanned aerial vehicles maneuver, distributed resource allocation and robot swarms, etc. One main objective in this context is to learn local controllers for  agents in a distributed manner so as to minimize the global  cost \cite{de2006decentralized}. 
For example, in the Linear Quadratic Regulator (LQR) control problem, the global objective is to minimize the sum of the local quadratic costs over all agents.

Nevertheless, the networked nature of MASs presents some unique challenges in designing distributed controllers. Observe that the agents are physically coupled with certain interconnections \cite{ding2019survey}, e.g., the buses in a microgrid are interconnected through structural links such as the power transmission lines.  Consequently, the controller synthesis at a bus has to account for the impact of other buses. To deal with the sophisticated coupling in MASs, the model-based distributed controller design has been studied in \cite{massioni2009distributed,antonelli2013interconnected,conte2016distributed}, where the interconnections among agents are modeled by a directed interaction graph to model the system dynamics. {However, these studies assume that the underlying system model is known, which may be infeasible in large-scale systems.}

To tackle the challenges in the MASs with unknown system models, data-driven approaches have emerged as a promising direction in learning local controllers.  Notably, data-driven Q-learning \cite{watkins1992q}, which is a model-free Reinforcement Learning (RL) approach \cite{bertsekas1995dynamic},  has been proposed to learn the optimal LQR controller online in the single agent case \cite{bradtke1994adaptive}. 
Motivated by this, some recent works (e.g., \cite{narayanan2016distributed,dizche2019sparse}) apply the Q-learning in the multi-agent LQR control and show that good performance can be achieved assuming that the knowledge of global state information is shared by a centralized coordinator \cite{narayanan2016distributed,dizche2019sparse}. Nevertheless, such a centralized coordinator is often not available in many scenarios.
It is therefore of great interest to study the MASs  where each agent can only learn state  information from neighbors with limited communication. Needless to say, this lacking of global state information inevitably makes the learning of the optimal controller more challenging,
calling for  distributed control based on \emph{partial observations}. 

In this work, we address the above problem by revisiting the  distributed LQR control problem with only partial observations of the global state. Specifically, we consider a distributed MAS where each agent has a discrete Linear Time Invariant (LTI) system with unknown dynamics. We assume that each agent can only share information with its neighbors over a communication graph. In particular, we focus on a more practical setting where the physical interconnection is different from the communication interconnection. This is often the case in the emerging cyber-physical systems, e.g., microgrid systems with distributed generators (DGs), the DGs are physically interconnected by a microgrid electric power network, and communicate in the cyberlayer \cite{bidram2014distributed}. 

Without careful design, the performance of distributed Q-learning can significantly degrade with only partial observations. To tackle this challenge, we propose a distributed Q-learning approach with a novel \emph{state tracking} strategy to facilitate the estimation of the global state through limited communication among neighboring agents. Intuitively, by exchanging  state estimations with neighboring agents, an individual agent would be able to improve its global state estimator as the information continuously diffuses across the network. Based on such a global state estimation, each agent then solves an \textit{approximate} Q-factor locally; and this learning process is carried out in parallel by all agents.

The main contributions of this work can be summarized as follows:

\begin{itemize}
	\item Considering distributed LQR control in MASs with only partial observations, we propose a novel distributed Q-learning approach with state tracking (ST-Q), where each agent first constructs a global state estimator based on local communication with its neighbors, and then solves an approximate Q-learning problem accordingly. 
	\item The convergence of distributed Q-learning algorithms in multi-agent LQR control has been underexplored. In this work, we fill this void and establish the convergence of the proposed distributed ST-Q algorithm.
	\item Compared with Q-learning under full observation \cite{narayanan2016distributed} and  Q-learning under partial observation only, our experimental results show that the proposed distributed ST-Q learning method achieves comparable performance with the full observation case. 
\end{itemize}

\section{Related Work}
{\bf{Distributed LQR Control.}} 
Distributed LQR control has recently garnered much attention. Notably, identical LTI system models across agents have been considered  which are coupled either in a global cost function  \cite{borrelli2008distributed} or in the state space \cite{vlahakis2019distributed}.
Since it is restrictive to assume identical systems for all agents, \cite{conte2016distributed} explores distributed model predictive control for heterogeneous LTI systems. 
Regarding the communication structure among agents, \cite{dong2010distributed}  requires all-to-all communications for the optimal control, and \cite{gorges2019distributed} assumes that agents share information with all their physically coupled neighbors. 
In this work, 
we consider a more challenging setting where (i) the system model parameters are unknown, and (ii) the communication topology is different from the system interconnection topology (cf. \cite{cheng2016gain}). 

{\bf{Multi-Agent Reinforcement Learning (MARL).}}
Taking a distributed Q-learning approach,
our focus is on MARL with partial observations (see,  e.g., Dec-POMDP \cite{bernstein2002complexity}). Information exchange is often utilized to facilitate the collaboration among agents. Notably, \cite{zhang2018fully} proposes a fully decentralized MARL where each agent shares local value function estimates with its neighbors to achieve a network-wide consensus.  In \cite{zhang2019distributed},  an approach is devised to enable each agent to communicate its local estimation of global optimal policy parameters with its neighbors. Note that both methods assume full  observation of the global state and control information to compute the gradient estimations. 
Assuming each agent solely has access to partial observations, 
a recent work \cite{li2019distributed} develops a policy gradient method where each agent shares an estimate of the global cost based on the local information only; further all agents act in parallel and have no control interaction with others. It is worth noting that the policy gradient method does not necessarily achieve the global optimum control policy \cite{carmon2018accelerated}.

Along a different line, an LQR control method using Q-learning is proposed \cite{bradtke1994adaptive}, which provides the first convergence guarantee on Q-learning based optimal control for the single agent case. A recent work \cite{alemzadeh2019distributed} presents a distributed Q-learning algorithm for coupled LTI systems with identical dynamics.  Assuming that the global state information is available through a central coordinator, \cite{narayanan2016distributed} establishes the convergence of distributed optimal controllers for coupled LTI systems. \cite{gorges2019distributed} studies the decentralized Q-learning, where the Q-function is calculated based on the local observations only and the system interconnection topology is assumed known. They use simulation studies to show that a near-optimal policy can be obtained. To fill the void, we propose a state tracking strategy to estimate the global state information at each agent based on its local information aggregated from neighbors. In such a way, the Q-factor can be solved more accurately at each agent by using the global state estimation. 

\section{Problem Formulation}
\subsection{Notation}
Throughout the paper, the set $\{1,2,\cdots,L\} \subseteq \mathbb{N}$ is denoted as $[L]$. The block-diagonal matrix $B$ with blocks $\{B_i\}_{i\in [L]}$ is denoted as $B = \textrm{diag}(B_1,B_2,\cdots,B_L)$. A column vector which stacks  subvectors $\{x_i\}_{i\in[L]}$ in a column is denoted as $X=\textrm{col}(x_1,x_2,\cdots,x_L)$. We use semicolon (;) to concatenate column vectors, hence $[x^{\top},u^{\top}]^{\top}=[x;u]$.  A graph is defined as $\mathcal{G}=(\mathcal{V},\mathcal{E})$, where  $\mathcal{V}$ is the set of nodes and $\mathcal{E} \subseteq \mathcal{V} \times \mathcal{V}$ is the set of edges connecting agents. For a node $i\in \mathcal{V}$, we denote $\mathcal{N}_i = \{ j \in \mathcal{V} | j\neq i, e_{ij}=(i,j) \in \mathcal{E}\}\cup \{i\}$ as the set of neighbors of node $i$ in the graph $\mathcal{G}$. We use $0_{n\times n}$ to represent a $n \times n$ zero matrix.

\subsection{Multi-Agent LTI System Model}
Consider a multi-agent network consisting of $L$ agents, where the LTI system dynamics at each agent 
$i \in [L]$ 
is given as follows:
\begin{equation}
\centering
x_i(t+1) = \textstyle \sum_{j=1}^{L}A_{ij}x_j(t) + B_i u_i(t)
\label{eqn:subsystem}
\end{equation}
where $x_i (t) \in \mathbb{R}^{n}$ is Agent $i$'s state vector and  $u_i (t) \in \mathbb{R}^{m}$ is its control input at time $t$. $A_{ij} \in \mathbb{R}^{n\times n}$ and  $B_i \in \mathbb{R}^{n \times m}$  are unknown system parameters. Putting  the system models across all agents in a more compact form, we have the following global system model:
\begin{equation}
X(t+1) = AX(t) + BU(t)
\label{eqn:globalsystem}
\end{equation}
where $X(t)=\textrm{col}(x_1(t),x_2(t),\cdots,x_L(t))$ is the global state vector, and $U(t)=\textrm{col}(u_1(t),u_2(t),\cdots,u_L(t))$ is the global control input vector. The global system matrix $A \in \mathbb{R}^{nL \times nL}$ is block-wise with entries $A_{ij}$ for each $i,j \in [L]$ and  $B = \textrm{diag}(B_1,B_2, \cdots, B_L)$. 

We further define a graph $\mathcal{G}^{d} = ([L],\mathcal{E}^{d})$ to model the interconnection topology, i.e., the state coupling, of the global system, where $[L]$ is the node (agent) set. Specifically, there exists an edge $e_{ij}^{d}\in\mathcal{E}^{d}$ between Agent $i$ and Agent $j$ if and only if they are interconnected, i.e., $A_{ij} \neq 0$.
Let $\mathcal{N}_i^{d}$ denote as the set of neighbors of Agent $i$ in the interconnection graph $\mathcal{G}^{d}$.

As is standard, we impose the following assumption on $(A, B)$  in this study.

\begin{restatable}[Stabilizability]{assumption}{asucontrollable}
	The system parameters $(A,B)$ in \eqref{eqn:globalsystem} are stabilizable.
	\label{asu:controllable}
\end{restatable}
Assumption~\ref{asu:controllable} indicates that there exists a control policy $\pi$ with $U(t) = \pi(X(t))$, such that the closed loop system $X(t+1) = AX(t)+B\pi(X(t))$ is asymptotically stable. 

\begin{figure}[t]
	\centering
	\includegraphics[width=0.43\linewidth]{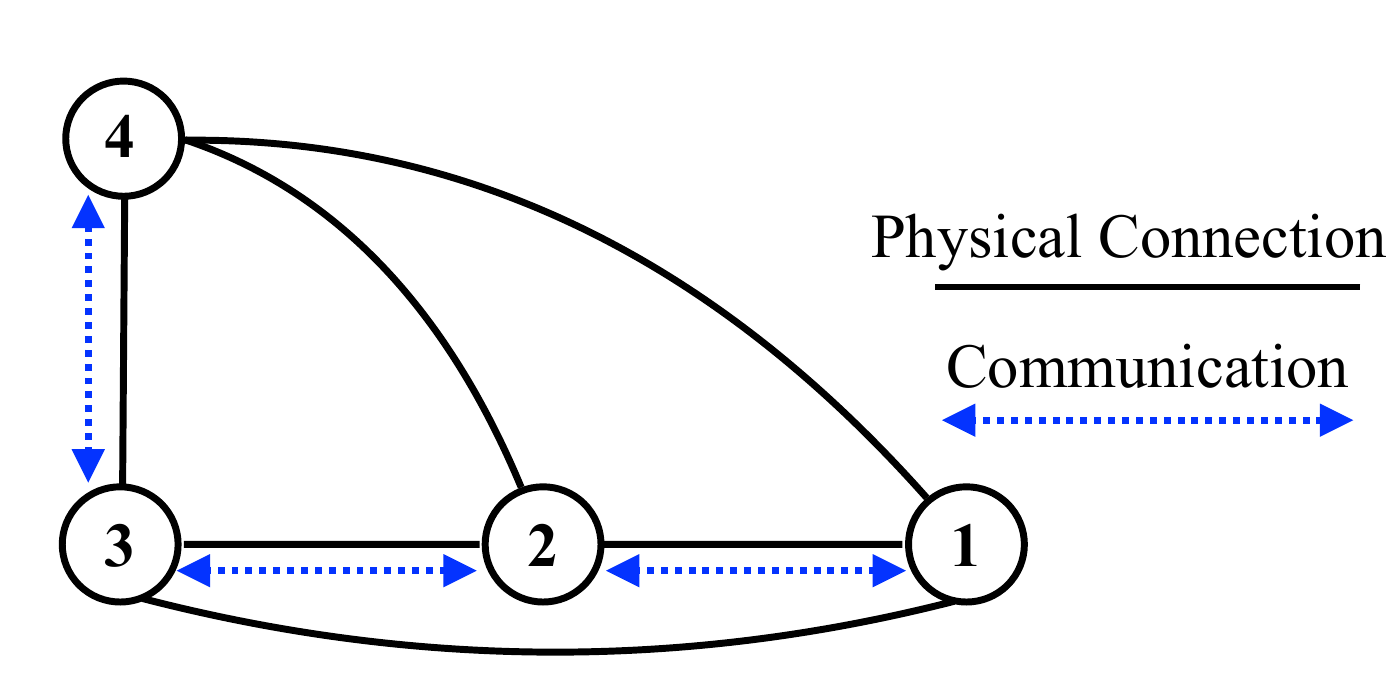}
	\caption{An example of the communication graph (dashed lines) and the interconnection graph (solid lines). All the agents are interconnected with physical connection, while the direct communication channel does not exist between neither Agent 4 and Agent 2 nor Agent 4 and Agent 1. }
	\label{fig:twograph}
\end{figure}


\subsection{Distributed LQR Control with Local Communication}

{\bf{Optimal distributed LQR control.}} For ease of expostion, we first present the optimal distributed LQR controller at each agent $i$ assuming  model parameters are known.
For the subsystem \eqref{eqn:subsystem} at each agent $i$, the  stage cost incurred by executing the control $u_i(t)$ in state $x_i(t)$ at time $t$ is given by
\begin{equation*}
\centering
g_i({x_i}(t), {u_i}(t))={x_i(t)^{\top}} {P_i} {x_i(t)}+{u_i}(t)^{\top} {R_i} {u_i(t)}
\end{equation*}
where $P_i \in\mathbb{R}^{n\times n}$ and $R_i\in\mathbb{R}^{m\times m}$ are positive semi-definite matrices. Let $J_i(x_i(0)) = \sum_{\tau=0}^{\infty}g_i(x_i(\tau),u_i(\tau))$ denote the local cost function at Agent $i$. The primary goal of the distributed LQR control is to minimize the sum of the local costs of all agents:
\begin{align}
\min_{\{u_i(\tau)\}}~\sum_{i=1}^{L}J_i(x_{i}(0)), ~~ \mathrm{s.t.}~\eqref{eqn:subsystem}.
\label{eqn:sumsubsystem}
\end{align}

Let $X(0)=\textrm{col}(x_1(0),x_2(0),\cdots,x_L(0))$, $P=\textrm{diag}(P_1, P_2, \cdots, P_L)$ and $R=\textrm{diag}(R_1, R_2, \cdots, R_L)$. It is clear that the distributed LQR control problem \eqref{eqn:sumsubsystem} is equivalent to the following LQR problem for the global system \eqref{eqn:globalsystem} with the initial state $X(0)$:
\begin{align}\label{globalLQR}
\min_{\{U(\tau)\}}~J(X(0)),
~~ \mathrm{s.t.}~\eqref{eqn:globalsystem},
\end{align}
with
\begin{align*}
J(X(0))
&={\sum}_{\tau = 0}^{\infty} X^{\top}(\tau)PX(\tau) + U^{\top}(\tau)QU(\tau).
\end{align*}

When the model parameters $A$ and $B$ are known, the optimal policy for \eqref{globalLQR} is given by linear feedback control, i.e.,
\begin{equation*}
U(t) = K^{*}X(t),
\end{equation*}
where $K^*=-(Q+B^TSB)^{-1}B^TSA$ and $S$ is the positive definite solution to the discrete Riccati equation:
\begin{align*}
S=A^TSA-A^TSB(Q+B^TSB)^{-1}B^TSA+P.
\end{align*}
The optimal control policy for each agent thus can be obtained as
\begin{align}\label{optimalcontroller}
u_i(t) = K_i^{*}X(t),
\end{align}
where the optimal controller $K_i^{*}$ is the $i$-th row of $K^*$. The LQR solution in this case can be efficiently computed via dynamic programming.
In the case when the model parameters $A$ and $B$ are unknown but each agent has a full observation of the global state $X(t)$, 
problem \eqref{eqn:sumsubsystem} can be solved efficiently by using reinforcement learning approaches  \cite{lewis2009reinforcement, narayanan2016distributed}. 

{\bf{Distributed LQR control with local communication.} } The primary focus of this paper is on distributed LQR control for a multi-agent network with unknown system parameters. 
Specifically, we define an undirected communication graph $\mathcal{G}^{c} = ([L],\mathcal{E}^{c})$ to model the information exchange in the multi-agent network. 
There exists an edge $e_{ij}^{c} \in \mathcal{E}^c$ between Agent $i$ and Agent $j$ if and only if they can communicate. 
Let $\mathcal{N}_i^{c}$ denote the set of neighbors of Agent $i$ in the communication graph $\mathcal{G}^{c}$. In particular,  we consider a general setting where the interconnection graph $\mathcal{G}^{d}$ and the communication graph $\mathcal{G}^{c}$ can be distinct, as illustrated in Figure \ref{fig:twograph}.
Since each agent does not have the full observation of the global state, it can only make its control decisions based on the local information, giving rise to the distributed LQR control that only relies on the partial observation (POD-LQR).

Let $x_{\mathcal{N}_i}(t) \in \mathcal{X}_{\mathcal{N}_i}$ denote the state information available for Agent $i$ at time $t$, which contains partial entries of the global state vectors. Agent $i$ then selects the local control input $u_i(t) \in \mathcal{U}_i$,  based on the information $x_{\mathcal{N}_i}(t)$ and a control policy $\tilde{\pi}_i$ with a linear feedback controller, i.e.,
\begin{align}
\tilde{\pi}_i: \mathcal{X}_{\mathcal{N}_i} \mapsto \mathcal{U}_i.
\label{eqn:localcontroller}
\end{align}
We further assume that $K_i$ is the feedback controller in the policy $\tilde{\pi}_i$. The goal of POD-LQR control is to find controllers that minimize the infinite horizon global cost  function $J(X(0))$:
\begin{align}
\min_{\{K_i\}}~J(X(0)) = {\textstyle\sum}_{i = 1}^{L}J_i(x_i(0)),~~\mathrm{s.t.}~ \textrm{\eqref{eqn:subsystem}, \eqref{eqn:localcontroller}}.   \label{eqn:localproblem}
\end{align}

In this work, we aim to achieve the optimal controller $K_i^*$ for each agent $i$ that is the same as in the case where the model parameters are known, by  solving Problem \eqref{eqn:localproblem} based on only a partial observation of the global state.

\section{Distributed Q-learning with State Tracking}

In this section, we propose a distributed Q-learning approach with state tracking to solve Problem \eqref{eqn:localproblem}, where each agent first constructs a global state estimator through communication with its neighbors, and then solves an approximate Q-learning problem locally using the state estimation. We start by presenting the preliminary on Q-learning with a full observation of the global state. Then, we present a state tracking scheme to facilitate the estimation of the global state through limited information exchange among neighboring agents.  Finally, the proposed state tracking based policy iteration algorithm is presented in detail.

\subsection{Global State based Q-Learning}\label{subsection:globalQ}

As mentioned earlier, Q-learning can be utilized to solve the distributed LQR control problem  \eqref{eqn:sumsubsystem} if each agent has a full observation of the global state $X(t)$. In what follows, we will briefly introduce the rationale behind Q-learning in the ideal case when the global state information is available at each agent.

Specifically, given the global state $X(t)$ and based on \eqref{optimalcontroller}, we consider the local control policy $\pi_i:~u_i(t)=K_i X(t)$ for some state feedback controller $K_i$. Then, the Q-factor for each agent $i$ can be defined as follows:
\begin{equation}
\centering
Q_i(x_i(t),u_i(t)) = g_i(x_i(t),u_i(t)) + J_i(x_i(t+1)),
\label{eqn:Qfunction}
\end{equation}
which gives the cumulative cost when agent $i$ starts from the state-control pair $(x_i(t),u_i(t))$ and follows the policy $\pi_i$ afterwards. Note that
\begin{equation*}
\centering
J_i(x_i(t))=Q_i(x_i(t),K_iX(t)).
\end{equation*}
The Bellman equation associated with the policy $\pi_i$ for the Q-factor can be written as
\begin{align}\label{bellmanQ}
&Q_i(x_i(t),K_i X(t))\nonumber\\
=& g_i(x_i(t),K_i X(t)) + Q_i(x_i(t+1),K_i X(t+1)),
\end{align}
and the corresponding Bellman optimality equation is
\begin{align}\label{bellmanQoptimal}
&Q^*_i(x_i(t),K^*_i X(t))\nonumber\\
=& g_i(x_i(t),K^*_i X(t)) + Q^*_i(x_i(t+1),K^*_i X(t+1)).
\end{align}
This implies that the optimal controller $K_i^*$ can be achieved as:
\begin{align*}
K_i^* = \arg \min_{K_i} Q^*_i(x_i(t),K_i X(t)).
\end{align*}
Therefore, to find the optimal controller $K_i^*$, it suffices to estimate the optimal Q-factor $Q^*_i$. And this can be achieved by using policy iteration where a sequence of  monotonically improved policies and Q-factors can be obtained, by running a policy evaluation step and then a policy improvement step in a recursive manner. For a better understanding of the Q-learning approach for LQR control, we start with the policy improvement step.

{\bf{Policy improvement.}} A key step in policy iteration is the policy improvement. Suppose we have determined the Q-factor $Q_i$ for a controller $K_i$ in the policy evaluation step. The policy improvement step aims to find a better controller:
\begin{align}\label{controllerupdate}
\centering
K_i^{\textrm{new}} = \arg \min_{K_i}(Q_i(x_i(t),K_i X(t))).
\end{align}
Note that the cost function $J_i$ is quadratic in the LQR control problem with a linear state feedback controller \cite{bertsekas1995dynamic}. Then, it can be shown that
\begin{equation*}
Q_i(x_i(t+1),K_iX(t+1))= x_i(t+1)^{\top}S_ix_i(t+1).
\end{equation*}
Here, $S_i$ is the cost matrix for the current controller $K_i$, which can be obtained by solving the discrete-time algebraic Riccati equation \cite{willems1971least}. Thus,  \eqref{eqn:Qfunction} can be rewritten as the following quadratic form:
\begin{align}\label{quadraticQ}
Q_i(x_i(t),u_i(t)) = [X(t);u_i(t)]^{\top}H_i[X(t);u_i(t)],
\end{align}
where $H_i$ is a symmetric block matrix defined as
\begin{align*}
\centering
H_i &=\begin{bmatrix}H_{i,11} & H_{i,12}\\H_{i,21}& H_{i,22}\end{bmatrix} =\begin{bmatrix}{\mathcal{A}_i}^{\top} {S_i \mathcal{A}_i}+\Tilde{P}_i &  {\mathcal{A}_i}^{\top} {S_i B_i}    \\  {B_i}^{\top} {S_i \mathcal{A}_i} &  {B_i}^{\top} {S_i B_i} + R_i\end{bmatrix}.
\end{align*}
Here, ${\mathcal{A}_i}$ is a row vector which stacks subvectors $\{A_{ij}\}_{j\in[L]}$ and $\Tilde{P}_i=\textrm{diag}(0_{n\times n}, \cdots, P_i,\cdots, 0_{n\times n})$ is a diagonal block matrix with the $(i,i)$-th block to be $P_i$. Based on the samples of the state-control pair $[X(t);u_i(t)]$, \eqref{controllerupdate} can be solved by using the first-order optimality condition:
\begin{equation}
K_i^{\textrm{new}} = -H_{i,22}^{-1}H_{i,21}.
\label{eqn:policyupdate}
\end{equation}
Summarizing, to obtain an improved controller $K_i^{\textrm{new}}$, it suffices to determine the Q-factor, in particular, the matrix $H_i$, in the policy evaluation step.

{\textbf{Policy evaluation.}} To determine the matrix $H_i$ in the policy evaluation step, along the same line as in \cite{bradtke1994adaptive}, we reformulate the quadratic form of $Q_i(x_i(t),u_i(t))$ in \eqref{quadraticQ} in a linear form parameterized by  parameter $\theta_i$:
\begin{equation}\label{linearQ}
Q_i(x_i(t),u_i(t)) = {y_i}(t)^{\top} {\theta_i},
\end{equation}
where ${y_i}(t)=[{x}_{1}^2(t),{x}_{1}(t){x}_{2}(t), \cdots, {x}_{L}(t)u_i(t),u_i^2(t)]$ is a vector containing all of the quadratic basis over the elements in $[X(t);u_i(t)]$, and ${\theta_i}$ is a vector in $\mathbb{R}^{(Ln+m)(Ln+m+1)/2}$. Here, the parameter ${\theta_i}$ is obtained through some manipulation after removing the redundant elements of the symmetric matrix $H_i$, i.e., the elements in the lower triangle of $H_i$. It is clear that in order to determine $H_i$, it suffices to determine the parameter $\theta_i$.

Based on the linear form \eqref{linearQ}, it is clear that the Bellman equation \eqref{bellmanQ} is equivalent to the following: 
\begin{equation}
g_i(x_i(t),u_i(t)) = ({y_i}(t)- {y_i}(t+1))^{\top}{\theta_i} \triangleq \phi_i(t)^{\top}{\theta_i},
\label{eqn:leastform}
\end{equation}
where $\phi_i(t) = {y_i}(t)- {y_i}(t+1)$. Note that $\phi_i(t)$ and the stage cost $g_i(x_i(t),u_i(t))$ can be known given the global state $X(t)$ and the control input $u_i(t)$. With sufficient samples of $(\phi_i(t), g_i(x_i(t),u_i(t)))$, $\theta_i$ can be obtained by solving a least square estimation problem.

\subsection{State Tracking}

It can be seen from \eqref{eqn:leastform} that the global state $X(t)$ is required to determine the parameter $\theta_i$ in the policy evaluation step, which however is not available in the POD-LQR control problem. To address this issue, we propose a state tracking scheme to facilitate the estimation of the global state $X(t)$ through the information exchange among agents over the communication graph $\mathcal{G}^{c}$.

More specifically, at time $t$  each agent $i$ maintains a local estimation ${Z}_i(t)$  of the global state $X(t)$: 
\begin{align*}
\centering
Z_i(t)= \textrm{col}(\bar{x}_{i1}(t),\bar{x}_{i2}(t),\cdots, \bar{x}_{iL}(t)),
\end{align*}
where $\bar{x}_{ij}(t)$ is  the estimation of Agent $j$'s state $x_j(t)$ at Agent $i$ for time $t$. In particular, $\bar{x}_{ii}(t)=x_i(t)$. Next, each agent communicates with and aggregates information from its neighbors in the communication graph $\mathcal{G}^{c}$ to update the local estimation ${Z}_i(t+1)$ as follows.

\textbf{Communication among neighboring agents.} At time $t+1$, the communication among agents includes two steps. First, each agent $i$ receives the state $x_j(t+1)$ from every neighbor $j\in\mathcal{N}_i^{c}$,
and then updates the corresponding entries in its  estimation $Z_i(t)$, i.e.,
\begin{equation*}
\centering
\bar{x}_{ij}(t) \to x_j(t+1),~\forall j \in \mathcal{N}_i^{c}.
\end{equation*}
Consequently, an updated estimation $\hat{Z}_i(t+1)=\textrm{col}(\hat{x}_{i1}(t+1),\hat{x}_{i2}(t+1),\cdots, \hat{x}_{iL}(t+1))$ can be obtained with
\begin{equation*}
\hat{x}_{ij}(t+1) =
\begin{cases}
\bar{x}_{ij}(t)&\forall j \notin \mathcal{N}_i^{c},\\
x_{j}(t+1) &\forall j \in \mathcal{N}_i^{c}.
\end{cases}
\end{equation*}
Next, each agent $i$ shares its updated global state estimation $\hat{Z}_i(t+1)$ with its neighbors in $\mathcal{G}^{c}$.

\textbf{Update of global state estimation.} After receiving the global state estimation $\hat{Z}_i(t+1)$ from the neighboring agents, Agent $i$ reconstructs a new estimation $Z_i(t+1)$ by aggregating all available information. In particular, for $j\in\mathcal{N}_i^{c}$, Agent $i$ has the accurate state information $x_j(t+1)$ of Agent $j$; for $j \notin \mathcal{N}_i^{c}$, Agent $i$ computes the state estimation $\bar{x}_{ij}(t+1)$ by taking a weighted average of the corresponding estimations $\hat{x}_{kj}(t+1)$ from its neighbors $k\in\mathcal{N}_i^{c}$. To model this `weighting' process, a doubly stochastic weight matrix, $W=[w_{ij}]\in\mathbb{R}^{L\times L}$, is used where $w_{ij}>0$ if and only if $(i,j)\in\mathcal{E}^{c}$. Otherwise, $w_{ij}=0$. The specific update rule is shown as following
\begin{equation}
\bar{x}_{ij}(t+1) =
\begin{cases}
\sum_{k=1}^{L} w_{ik} \hat{x}_{kj}(t+1)&\forall j \notin \mathcal{N}_i^{c},\\
x_j(t+1) &\forall j \in \mathcal{N}_i^{c}.
\end{cases}
\label{eqn:updateState}
\end{equation}

\begin{figure}
	\centering
	\includegraphics[width=.5\columnwidth]{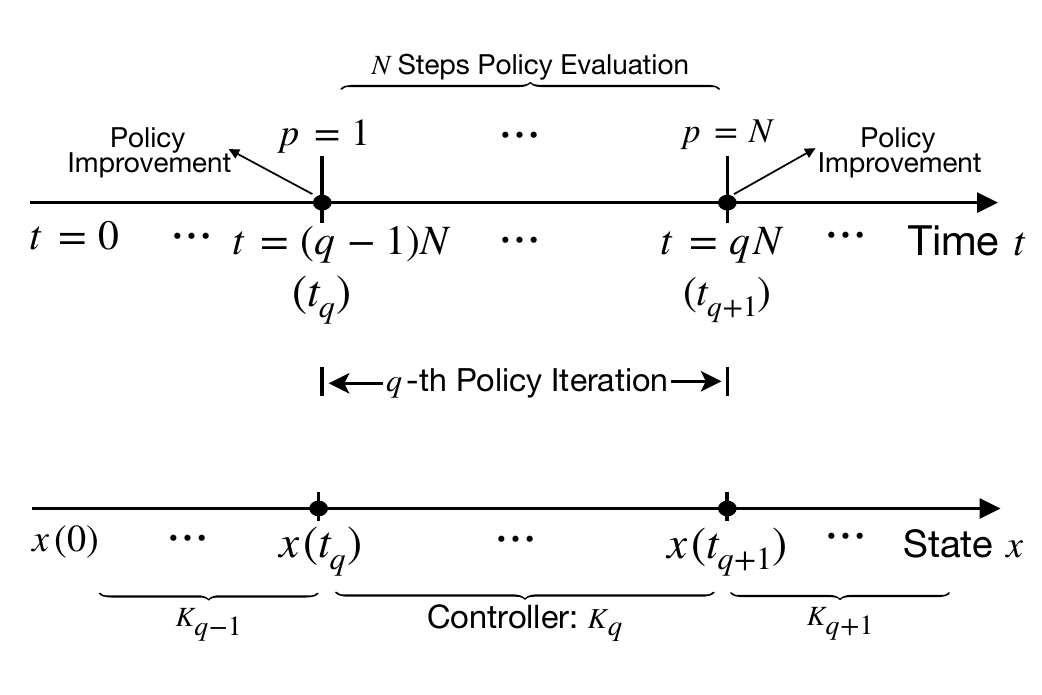}
	\caption{Illustration of the time scales in Algorithms~(\ref{alg:gst}) and (\ref{alg:qlearning}): $t$ denotes the time steps, and $t_q$ denotes the time instance for the $q$-th policy iteration. Policy evaluation is carried out $N$ times within each policy iteration. }
	\label{fig:timescale}
\end{figure}
\renewcommand\theContinuedFloat{\alph{ContinuedFloat}}
\begin{algorithm}[t]
	\ContinuedFloat*
	\caption{ST based Policy Evaluation (ST-E)}
	\begin{algorithmic}[1]
		\Require $K_{iq}$: evaluation controller, $\eta_i(t)$:  excitation noise.
		\State $p=1$.
		\For{$p=1,\cdots,N$}
		\State Apply $u_i(t) = -K_{iq}Z_i(t) + \eta_i(t)$.
		\State Measure $x_i(t+1)$ and receive $x_j(t+1)$, $j\in \mathcal{N}_i^{c}$.
		\State Receive $\hat{Z}_j(t+1)$ from all $j \in \mathcal{N}_i^{c}$ and update $Z_i(t+1)$ following \eqref{eqn:updateState}.
		\State Obtain $u_i(t+1) = -K_{iq}Z_i(t+1)$.
		\State Update $\hat{\theta}_{iq}(p)$ using \eqref{eqn:gradient}.
		\State Set $p=p+1$ and $t=t+1$.
		\EndFor
		\State \textbf{Return} $\hat{\theta}_{iq}$.
	\end{algorithmic}
	\label{alg:gst}
\end{algorithm}
\begin{algorithm}[h!]
	\ContinuedFloat
	\caption{ST based Q-learning (ST-Q)}
	\begin{algorithmic}[1]
		\Require $K_{i1}$: initial stable controller, $\theta_{i1}(0)=0$: initial estimation, $q=1$: policy improvement index, $t=0$: time index, $\varepsilon_{K}$: tolerance error, $x_i(0)$: initial state, $Z_i(0)$: initial global state estimation.
		\Repeat
		\For{Agent $i=1,\cdots,L$}
		\State Estimate $\theta_{iq}$ through Algorithm~(\ref{alg:gst}).
		\EndFor
		\For{Agent $i=1,\cdots,L$}
		\State Obtain $H_{iq}$ from $\hat{\theta}_{iq}(N)$.
		\State Update policy $K_{i(q+1)} = -H_{iq,22}^{-1}H_{iq,21}$.
		\State Set $\hat{\theta}_{i(q+1)}(0) = \hat{\theta}_{iq}(N)$.
		\EndFor
		\State Set $q=q+1$.
		\Until{$\| \hat{\theta}_{i(q+1)} - \hat{\theta}_{iq}\| < \varepsilon_{K}$, $\forall i \in [L]$}
	\end{algorithmic}
	\label{alg:qlearning}
\end{algorithm}

\subsection{ST-based Policy Iteration for Q-learning}

Based on the estimation $Z_i(t)$ of the global state $X(t)$ achieved by state tracking, each agent $i$ now is able to carry out an approximate Q-learning locally by using policy iteration to solve the POD-LQR control problem \eqref{eqn:localproblem}. As mentioned earlier in Section \ref{subsection:globalQ}, the policy iteration includes two main steps, i.e., the policy evaluation and the policy improvement step. We summarize the important steps below, and more details can be found in Algorithms~(\ref{alg:gst}) and (\ref{alg:qlearning}).

Specifically, each agent $i\in [L]$ first starts with a stabilizing initial controller $K_{i1}$, and iteratively runs the policy evaluation and the policy improvement as shown in Fig.~\ref{fig:timescale}. At iteration $q$, 
\begin{itemize}
	\item\textit{Policy Evaluation.} 
	In this step, each agent $i$ aims to determine the Q-factor $Q_{iq}$ for a given controller $K_{iq}$. As shown in \eqref{linearQ}, it suffices to determine the parameter $\theta_{iq}$. To this end, we resort to least square estimation  based on \eqref{eqn:leastform} with $N$ samples per policy evaluation step:
	\begin{align}\label{lse}
	\centering
	\min_{\theta_{iq}}~ \sum\nolimits_{t=t_q}^{t_{q}+N-1} \|\bar{\phi}_i(t)^T\theta_{iq}-g_i(x_i(t),u_i(t))\|^2,
	\end{align}
	where $\bar{\phi}_i(t) = {\bar{y}_i}(t)- {\bar{y}_i}(t+1)$ and ${\bar{y}_i}(t)=[\bar{x}_{i1}^2(t),\bar{x}_{i1}(t)\bar{x}_{i2}(t),\allowbreak \cdots, \bar{x}_{iL}(t)u_i(t),u_i^2(t)]$ is the vector containing all the quadratic basis over the elements in the estimated global information vector $[Z_i(t);u_i(t)]$. Here,  $u_i(t)=-K_{iq}Z_i(t)+\eta_i(t)$ where $\eta_i(t)$ is the input noise to ensure that the system at each agent is persistently excited. For ease of exposition, we denote $g_i(x_i(t),u_i(t))$ as $g_i(t)$, and consider $(\bar{\phi}_i(t),g_i(t))$ as a sample for the least square estimator. 
	
	To solve the least square estimation problem \eqref{lse}, an online gradient descent method is run for $N$ iterations: At each iteration $p\in [1,N]$, each agent (i) constructs the global state estimation $Z_i(t)$ and $Z_i(t+1)$ so as to obtain a sample $(\bar{\phi}_i(t),g_i(t))$
	as shown in Algorithm~(\ref{alg:gst}), and (ii) updates the estimation of $\theta_{iq}$ by using gradient descent with a learning rate $\alpha$:
	\begin{equation}
	\centering
	\textstyle \hat{\theta}_{iq}(p+1) = \hat{\theta}_{iq}(p) - \alpha\bar{\phi}_i(t)\big(\hat{\theta}_{iq}(p)^{\top}\bar{\phi}_i(t) - g_i(t) \big), 
	\label{eqn:gradient}
	\end{equation}
	where $\hat{\theta}_{iq}(p)$ is the estimation of $\theta_{iq}$ at policy evaluation step $p$.
	
	\item\textit{Policy Improvement.} Given $\hat{\theta}_{iq}=\hat{\theta}_{iq}(N)$ obtained in the policy evaluation step, each agent is able to reconstruct the matrix $\hat{H}_{iq}$, so that the controller can be updated as follows:
	\begin{align*}
	{K}_{i(q+1)} = -\hat{H}_{iq,22}^{-1}\hat{H}_{iq,21}.
	\end{align*}
\end{itemize}

This policy iteration procedure stops if the following condition is satisfied:
\begin{equation*}
\|\hat{\theta}_{i(q+1)} - \hat{\theta}_{iq}\|<\varepsilon_K, \forall i\in [L]
\end{equation*}
where $\varepsilon_K$ is a predefined threshold for the estimation error.


\section{Convergence Analysis}\label{section:convergence}
In this section, we establish the convergence of the proposed ST-Q learning algorithm. To this end, we first make a few standard assumptions  for multi-agent reinforcement learning \cite{nedic2018distributed,doan2019finite}.
\begin{restatable}[Communication Connectivity]{assumption}{asuconnect}
	The communication graph $\mathcal{G}^{c}$ is connected and static.
	\label{asu:connect}
\end{restatable}
\begin{restatable}[Weight Matrix]{assumption}{asuweightmatrix}
	There exists a positive constant $\eta$ such that the weight matrix $W=[w_{ij}]\in \mathbb{R}^{L \times L}$ is doubly stochastic and  $w_{ij} \geq \eta$ if $j \in \mathcal{N}_i^{c}$. In particular, $w_{ii} \geq \eta$, $\forall i \in [L]$. 
	\label{asu:weightmatrix}
\end{restatable}

\begin{restatable}[Decaying Excitation  Noise]{assumption}{asudecaynoise}
	The input noise $\eta(t)$ is with the decaying factor $\upsilon(p)$, 
	\begin{equation*}
	\begin{aligned}
	\upsilon(p) = c^{p},&\quad 0<c<1,\\
	\eta(t) = \upsilon(p)\beta(t),&\quad t_q \leq t < t_{q+1}.
	\end{aligned}
	\end{equation*}
	The input noise for the global system is $E(t) = \Upsilon(p) \boldsymbol{\beta}(t)$. The system is further persistently excited with the input noise, i.e., $\forall i\in [L]$, $\forall q$:
	\begin{equation*}
	mI \leq  \textstyle\sum_{t=t_q}^{t_{q}+N-1}\phi_i(t)\phi_i^{\top}(t) \leq MI,
	\end{equation*}
	where $0< m \leq M < \infty$.
	\label{asu:decaynoise}
\end{restatable}

\begin{restatable}[Step Size of Gradient Descent]{assumption}{asulr} The step size in \eqref{eqn:gradient} is fixed and satisfies: $0< \alpha < 1/M$.
	\label{asu:lr}
\end{restatable} 
Assumptions \ref{asu:connect} and \ref{asu:weightmatrix} are imposed to facilitate the information diffusion across the network. The excitation condition in Assumption \ref{asu:decaynoise} is to guarantee the convergence of policy evaluation. Along the same line as in \cite{goodwin2014adaptive,bradtke1994adaptive}, this condition can be met by adding sinusoidal noise of various frequencies to $u_i(t)$. Moreover, we further assume that the  input noise $\eta_i(t)$ is decaying for a more accurate system state estimation as the algorithm gradually converges.

For convenience, we restate in the following lemma the convergence  result on  distributed Q-learning with full global state observations \cite{bradtke1994adaptive,narayanan2016distributed}.

\begin{restatable}[Convergence of Distributed Q-learning with Full Observation]{lemma}{lemmafullobservation} 
	Suppose that Assumptions \ref{asu:controllable}, \ref{asu:decaynoise}, \ref{asu:lr} are satisfied, and $K_{i1}$ is a stabilizing controller. There exists $N<\infty$ such that the sequence of stabilizing controllers $\{K_{iq}\}_{q=1}^{\infty}$ generated by the  Q-learning Policy Iteration mechanism with global state information converges, i.e., $\forall i \in [L]$:
	\begin{equation*}
	\lim_{q \rightarrow \infty} \| {K}_{iq} -  {K}_i^{*}\| = 0,
	\end{equation*}
	where ${K}_i^{*}$ is the optimal feedback controller. 
	\label{lemma:fullobservation}
\end{restatable}
\begin{proof}[Proof Sketch]
	First, we show that by replacing recursive least squares (RLS) with stochastic gradient descent (SGD) in the adaptive policy iteration algorithm proposed in \cite{bradtke1994adaptive}, the policy iteration algorithm also generates a sequence of stabilizing controls converging to the optimal in the single agent case. Under the setting  when agents have full observation of the global state, \cite{narayanan2016distributed} considers a policy iteration algorithm with the RLS estimation method and provides the convergence proof by utilizing the result in \cite{bradtke1994adaptive}. Following the same line in \cite{narayanan2016distributed} and the result obtained in the preceding step, this lemma can be proved. The full proof is in Appendix \ref{app:A}.
\end{proof}

To establish the convergence of the proposed ST-Q learning approach, we first evaluate the estimation error of $\hat{\theta}_{iq}$ in Algorithm~(\ref{alg:gst}) with respect to $\theta_{iq}$ by characterizing the convergence performance of the global state estimation obtained by state tracking. 

\begin{restatable}[Convergence of Parameter Estimation]{lemma}{lemmaestimationerror}
	Under Assumptions \ref{asu:controllable}-\ref{asu:lr}, there exists $N<\infty$, such that 
	
	\begin{itemize}
		\item[(a)] the global state estimation error is bounded above by some arbitrarily small $\delta>0$, i.e., $\forall i\in [L]$, $\forall q$:
		\begin{equation*}
		\|{Z}_i(t_q+N)- {X}(t_q+N)\| \leq \delta,
		\end{equation*}
		\item[(b)] the estimation error of $\theta_i$ in \eqref{eqn:leastform} is bounded above by some arbitrarily small $\xi>0$ when $q$ is large enough, i.e., $\forall i \in [L]$:
		\begin{equation*}
		\|\theta_{iq} - \hat{\theta}_{iq}\|\leq \xi,
		\end{equation*}
		where $\hat{\theta}_{iq}$ is an estimate obtained by the ST-based approach and $\theta_{iq}$ is obtained with full observations. Note that  $\hat{\theta}_{iq}= \hat{\theta}_{iq}(N)$, ${\theta}_{iq}= {\theta}_{iq}(N)$.
	\end{itemize}
	\label{lemma:estimationerror}
\end{restatable}
\begin{proof}[Proof Sketch]
	(a) First define  $\allowbreak \epsilon_{ik}(t) =  \sum_{j \in \mathcal{N}_k} w_{ij} ({x}_{k}(t)-{x}_{k}(t-1)) $ and $ \bar{x}_{av,k}(t) = \frac{1}{L} \sum_{j=1}^{L}\bar{x}_{jk}(t)$. The global state estimation error can be shown as 
	\begin{equation*}
	\begin{aligned}
	\|{Z}_i(t)- {X}(t)\| =&  \sqrt{\sum_{k =1}^{L} \|x_k(t) - \bar{x}_{ik}(t)\|^2_2 }\\
	\leq&  \sum_{k =1}^{L} {\|x_k(t) - \bar{x}_{ik}(t)\| }\\
	\leq&  {  \sum_{k =1}^{L}  \|x_{k}(t) - \bar{x}_{av,k}(t)\| } + {  \sum_{k =1}^{L}\|\bar{x}_{av,k}(t)- \bar{x}_{ik}(t)\| },\\
	\end{aligned}
	\end{equation*}
	where the first term can be analyzed by bringing in the definition of  $\bar{x}_{av,k}$ and using the stability property of the controller. The second term is analyzed by formulating a perturbed consensus problem following the same line as in \cite[Lemma 3]{nedic2018distributed}:
	\begin{equation*}
	\begin{aligned}
	\overline{x}_{ik}(t) &=  \textstyle\sum_{j} w_{ij} \bar{x}_{jk}(t-1) + \epsilon_{ik}(t),\\
	\epsilon_{ik}(t) &=  \textstyle\sum_{j \in \mathcal{N}_k} w_{ij} ({x}_{k}(t)-{x}_{k}(t-1)).\\
	\end{aligned}
	\end{equation*}
	
	(b) Recall the $p$-th gradient descent step: 
	\begin{equation*}
	\begin{aligned}
	\textstyle \theta_{iq}(p+1) &= \theta_{iq}(p) - \alpha\phi_i(t)\cdot\big(\theta_{iq}(p)^{\top}\phi_i(t) - g_i(t) \big) , \\  
	\textstyle \hat{\theta}_{iq}(p+1) &= \hat{\theta}_{iq}(p) - \alpha\hat{\phi}_i(t)\cdot\big(\hat{\theta}_{iq}(p)^{\top}\hat{\phi}_i(t) - \hat{g}_i(t) \big).  \\  
	\end{aligned}
	\end{equation*}
	For convenience, define
	\begin{equation*}
	\begin{aligned}
	\Phi_i(t_q+N-\tau) &\triangleq I-\alpha\phi_i(t_q+N-\tau)\phi_i^{\top}(t_q+N-\tau),\\
	\Pi_i(N) &\triangleq \textstyle\prod_{\tau=1}^{N}\Phi_i(t_q+N-\tau),\\
	G_i(t_q+N-\tau)& \triangleq \alpha\phi_i(t_q+N-\tau) g_i(t_q+N-\tau).\\ 
	\end{aligned}
	\end{equation*}
	By using the deduction of $\theta_{iq}(p)$, we obtain the relationship between $\theta_{iq}=\theta_{iq}(N)$ and $\theta_{i(q-1)}=\theta_{iq}(0)$, as follows: 
	\begin{equation*}
	\centering
	\begin{aligned}
	\theta_{iq} &= 
	\Pi_i(N) \theta_{i(q-1)} + \sum_{\tau=2}^{N} \Pi_i(\tau-1)G_i(t_q+N-\tau) + G_i(t_q+N-1),\\
	\hat{\theta}_{iq} &= 
	\hat{\Pi}_i(N) \hat{\theta}_{i(q-1)} + \sum_{\tau=2}^{N}\hat{\Pi}_i(\tau-1) \hat{G}_i(t_q+N-\tau) +  \hat{G}_i(t_q+N-1).\\
	\end{aligned}
	\end{equation*}
	It follows that the estimation error of $\|\hat{\theta}_{iq}-\theta_{iq}\|$ can be obtained by analyzing $\|\Phi_i-\hat{\Phi}_i\|$, $\|\Pi_i-\hat{\Pi}_i\|$ and $\|G_i-\hat{G}_i\|$ using the result from Lemma \ref{lemma:estimationerror} (a). 
	
	The full proof of statements (a) and (b) is relegated to Appendices \ref{app:C} and \ref{app:D}, respectively.
\end{proof}

Based on Lemma \ref{lemma:estimationerror}, we are now able to characterize the convergence performance of the ST based  learning.

\begin{restatable}[Convergence of the ST-Q learning]{theorem}{theoremfinal}
	Suppose Assumption \ref{asu:controllable}-\ref{asu:lr} are satisfied and $K_{i1}$ is a stabilizing controller. Then, for any  $\varepsilon_K>0$, there exist $N<\infty$ and $q<\infty$, such that the ST-based Q-learning mechanism described in Algorithms~(\ref{alg:gst}) and (\ref{alg:qlearning}) generates a sequence of stabilizing controllers  $\{\hat{K}_{iq}\}_{q=1}^{\infty}$ that converge to the optimal controller, i.e.,  $\forall i \in [L]$:
	\begin{equation*}
	\| \hat{K}_{iq} -  {K}_i^{*}\| \leq \varepsilon_K.
	\end{equation*}
	\label{theorem:final}
\end{restatable}
\begin{proof}
	
	By Lemma \ref{lemma:estimationerror}, there exist $N < \infty$ and $q<\infty$ such that,
	\begin{align*}
	\| \hat{\theta}_{i(q-1)} -   {\theta}_{i(q-1)} \| \leq \xi.
	\end{align*}
	
	Following \cite{bradtke1994adaptive}, there exists a constant $k_0 > 0$, such that,
	\begin{align*}
	\| \hat{K}_{i(q)} -   {K}_{i(q)} \| \leq k_0  \|\hat{\theta}_{i(q-1)} - {\theta}_{i(q-1)}\|.
	\end{align*}
	
	Hence, we obtain that
	\begin{align*}
	\| \hat{K}_{iq} -   {K}_{iq} \| \leq    k_0\xi.
	\end{align*}
	
	Besides, from Lemma \ref{lemma:fullobservation}, there exists $q < \infty$, such that
	\begin{equation*}
	\| {K}_{iq} -   {K}_{i}^{*} \| \leq \xi_k .
	\end{equation*}
	
	By using the triangle inequality, we obtain that
	\begin{align*}
	\centering
	\| \hat{K}_{iq} -  {K}_{i}^{*} \| =  \| \hat{K}_{iq} - {K}_{iq} +{K}_{iq} -  {K}_{i}^{*} \| &\leq k_0\xi+\xi_k \triangleq \varepsilon_K.    
	\end{align*}
\end{proof}

\begin{figure}[t]
	\centering
	\begin{subfigure}[b]{0.42\columnwidth}
		\centering
		\caption{Global controllers.}
		\includegraphics[width=\textwidth]{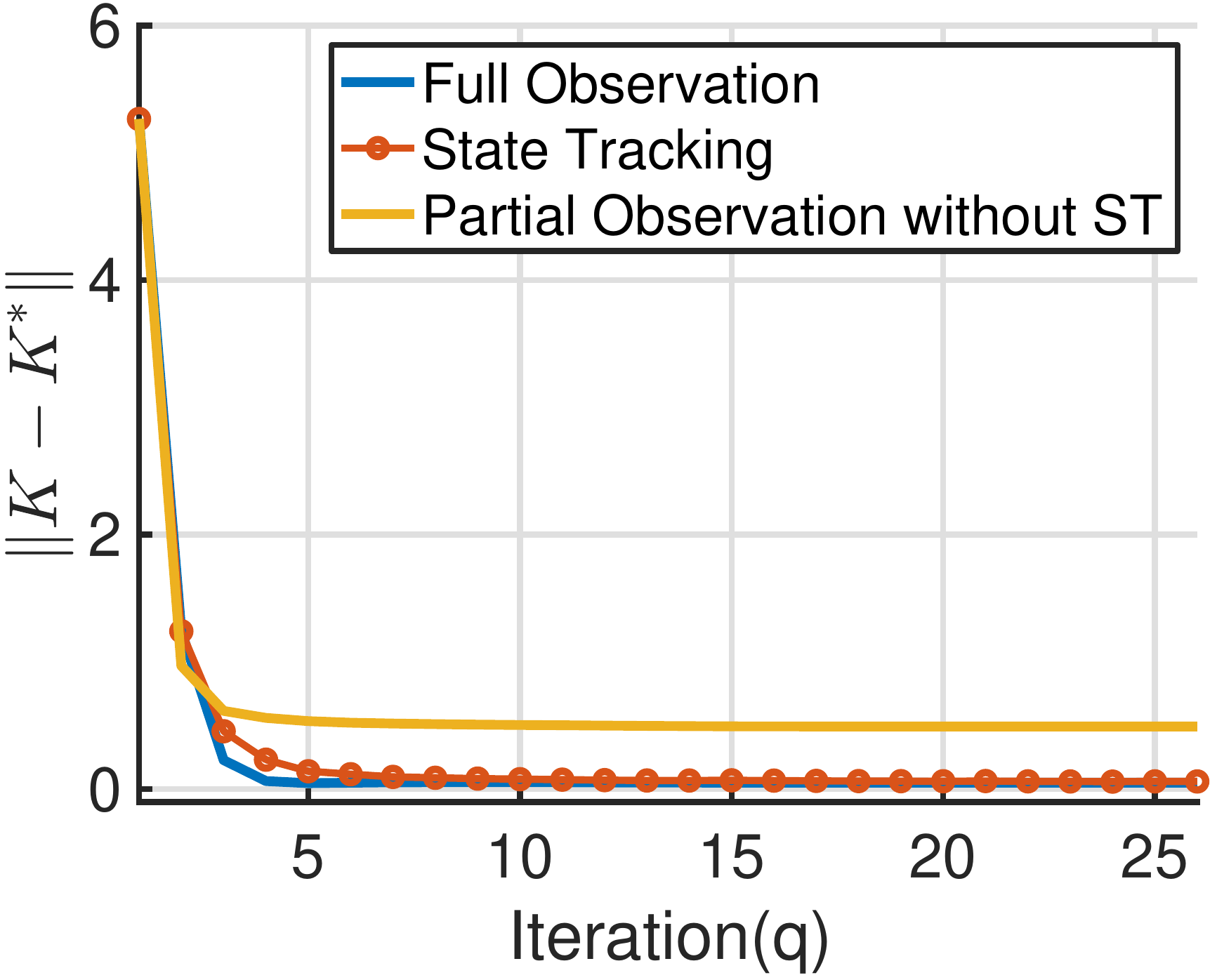}
		\label{fig:compareconvergence}
	\end{subfigure}
	\hfill
	\begin{subfigure}[b]{0.42\columnwidth}
		\centering
		\caption{Local controllers.}
		\includegraphics[width=\textwidth]{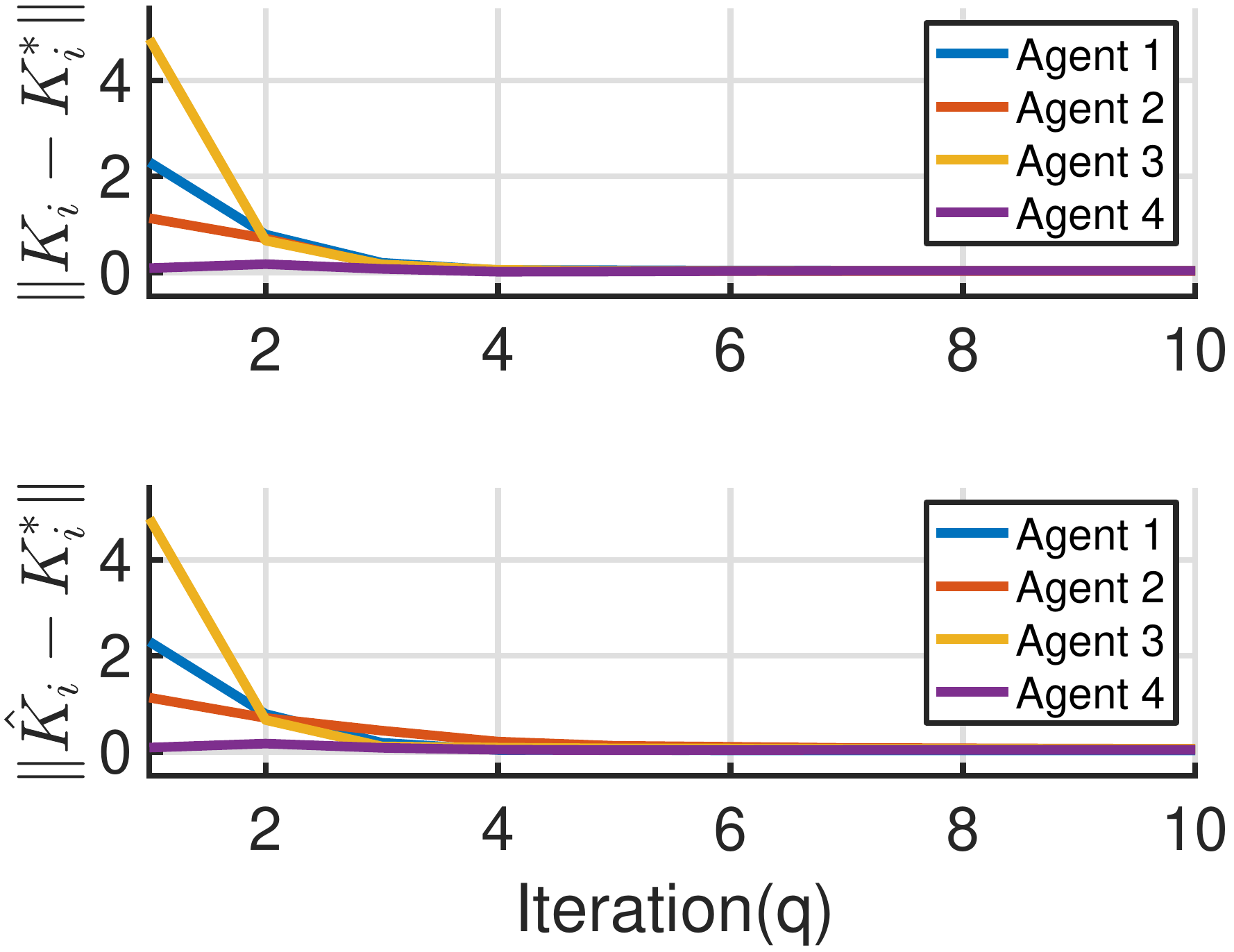}
		\label{fig:rlscompare}
	\end{subfigure}
	\hfill
	\caption{Convergence comparisons among three cases: ST based Q-learning, partial observation with no ST and full observation Q-learning. The controller obtained in three cases are denoted as $\hat{K}$, ${K}_{\textrm{partial}}$ and  ${K}$, respectively. The optimal controller is denoted as  $K^{*}$. The upper figure in Figure~(\ref{fig:rlscompare}) is the local controllers convergence behavior with full observation while the lower one is the result with ST strategy.}
	\label{fig:convergence}
\end{figure}
\begin{figure}[t]
	\centering
	\begin{subfigure}[b]{0.42\columnwidth}
		\centering
		\caption{SGD with large step size.}
		\includegraphics[width=\textwidth]{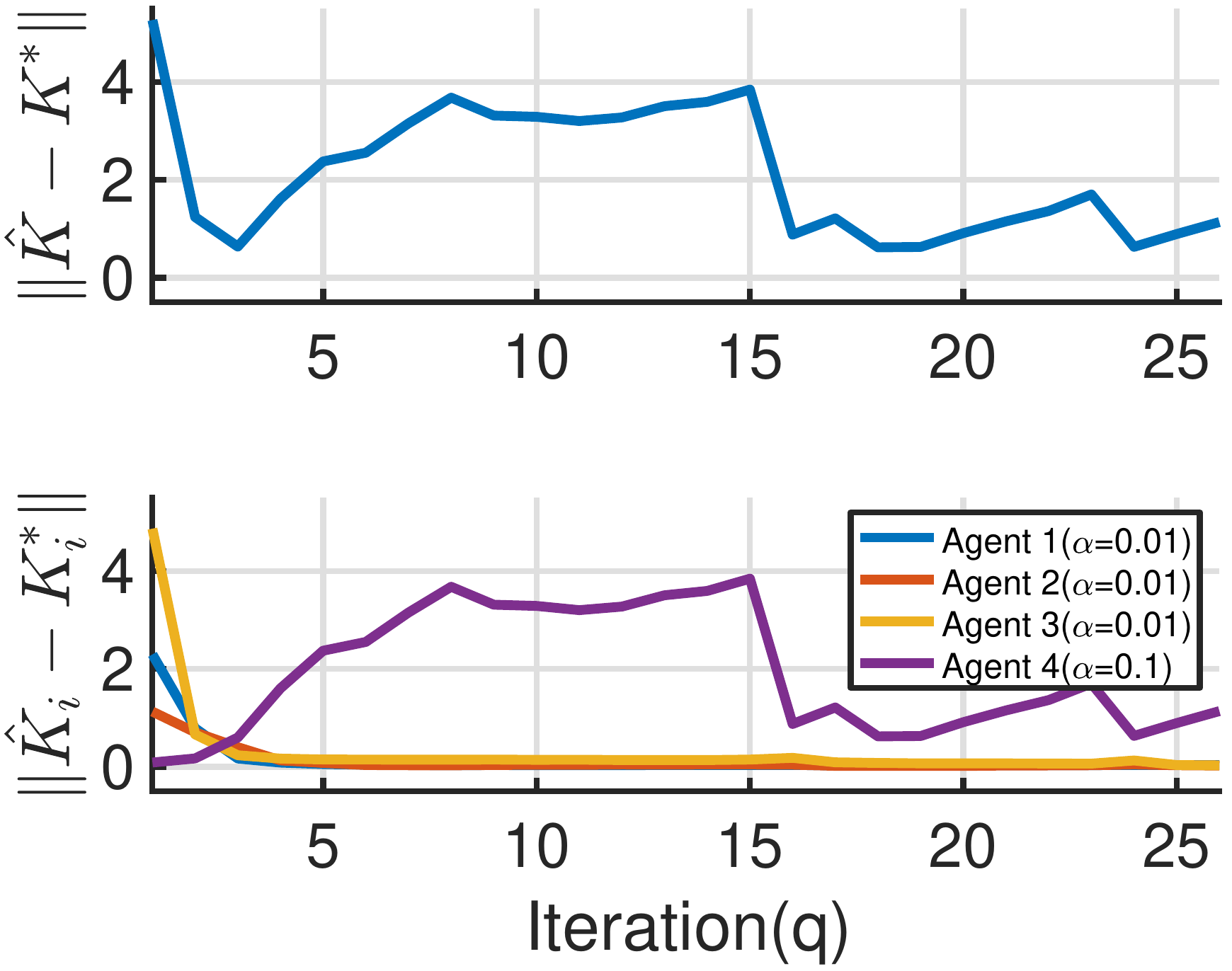}
		\label{fig:largeA}
	\end{subfigure}
\hfill
	\begin{subfigure}[b]{0.42\columnwidth}
		\centering
		\caption{SGD with small step size.}
		\includegraphics[width=\textwidth]{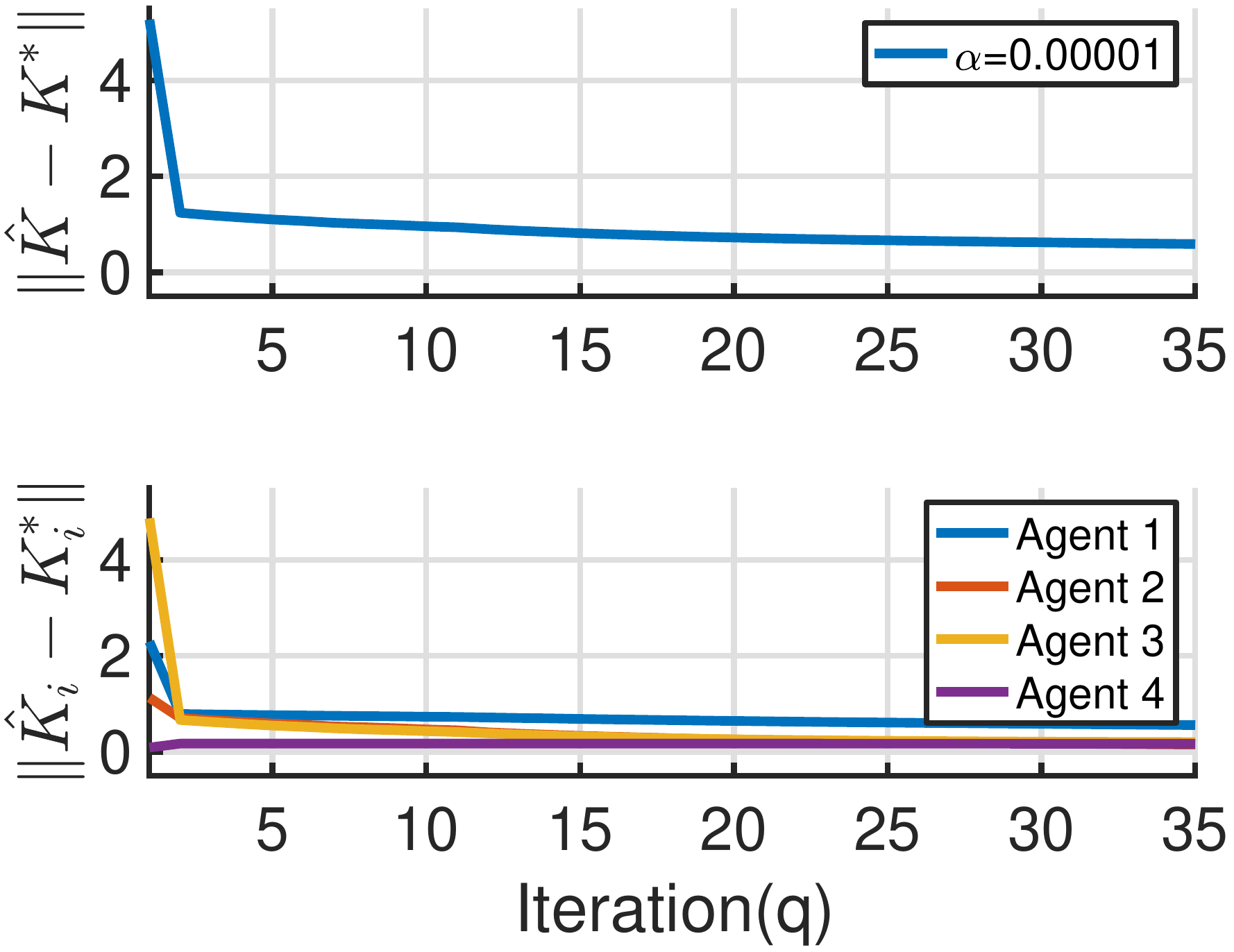}
		\label{fig:smallA}
	\end{subfigure}  
	\caption{Comparison of different SGD step size in the ST based Q-learning. The upper figures in Figures~(\ref{fig:largeA}) and (\ref{fig:smallA}) show the convergence behavior of the global controller while the lower figures are of the local controllers.}
	\label{fig:stepsize}
\end{figure}
\section{Experiments}
In this section, we first introduce the experimental setup, and then evaluate the performance of the proposed ST-Q learning method. In particular, 
we compare our approach with two baselines:  (i) distributed Q-learning with global state (DQG), and (ii) distributed Q-learning with partial observation of the global state (DQP) where the absent state information is set as 0, i.e., $ x_j(t) = 0, \forall j\notin \mathcal{N}_i^c$. We further examine the impact of the hyper-parameters, including the step size $\alpha$, interval $N$ and the excitation noise $\eta$, on the performance of our approach to verify the assumptions made in Section \ref{section:convergence}.   

\subsection{Experimental Setup}
We consider the following global system with four agents, and the communication topology and interconnection topology as demonstrated in Fig. \ref{fig:twograph}:	
\begin{equation*}
\begin{bmatrix}x_1(t+1) \\ x_2(t+1) \\ x_3(t+1) \\ x_4(t+1)\end{bmatrix} = A\begin{bmatrix}x_1(t) \\ x_2(t) \\ x_3(t) \\ x_4(t)\end{bmatrix} + B \begin{bmatrix}u_1(t) \\ u_2(t) \\ u_3(t) \\ u_4(t)\end{bmatrix},
\end{equation*}
where the initial state for each agent is given as follows
\begin{equation*}
x_i(0) = 0.01,~i=1,2,3,4.
\end{equation*}
Note that the parameters can be chosen arbitrarily as long as they meet the assumptions in Section \ref{section:convergence}. The system parameters $A$ and $B$ are stabilizable and are set as
\begin{equation*}
\centering
A = \left[\begin{array}{cccc}0.2 & 0.4 & 0.1 & 0.01 \\ 0.4 & 0.2 & 0.3 & 0.1 \\ 0.1 & 0.3 & 0.3 & 0.4\\0.2 & 0.1 & 0.5 & 0.3\\\end{array}\right],~B  = \left[\begin{array}{cccc}1 & 0 & 0 & 0\\ 0 & 1 & 0 & 0 \\  0 & 0 & 1 & 0\\ 0 & 0 & 0 & 1 \end{array}\right].
\end{equation*}
The weighting matrices for the LQR problem are selected as $P_i=R_i=1$, for $i=1,2,3,4$. By solving the discrete Riccati equation, the optimal controller is obtained as,
\begin{equation*}
\centering
K^{*} = \begin{bmatrix}
0.1223 &0.2279&0.0779&0.0251\\
0.2267 &0.1279&0.1823&0.0714\\
0.0796 &0.1869&0.1944&0.2341\\
0.1212 &0.0742&0.2838&0.1756\\
\end{bmatrix}.
\end{equation*}
Initial stable controller for Algorithms~(\ref{alg:gst}) and (\ref{alg:qlearning}) is chosen to be:
\begin{equation*}
\centering
K_{1} = \begin{bmatrix}
1 &1&0.0004&2\\
1 &0.2&1&0.1\\
4 &0.1&1&3\\
0.2 &0.1&0.3&0.2\\
\end{bmatrix}.
\end{equation*}
The weight matrix for the communication graph $\mathcal{N}_i^c$ is:
\begin{equation*}
\centering
W = \begin{bmatrix}
0.5 &0.5&0&0\\
0.5 &0.3&0.2&0\\
0 &0.2&0.2&0.6\\
0 &0&0.6&0.4\\
\end{bmatrix}.
\end{equation*}
The experiments are carried out with $N=1000$ and step size $\alpha = 0.01$. Follow the approach in \cite{goodwin2014adaptive}, the excitation noise $\eta_i(t)$ for Agent $i$ is designed as $\big(b_i\cdot \textrm{rand}(-1,1)+a_i \cdot \sum_{\omega=1}^{15}\sin(\omega t)^3\cos(\omega t)\big)\cdot \upsilon(p)$, where the decaying factor $\upsilon(p)=0.9999^p$ and $a_i,b_i \geq 0$ are constants. 

\subsection{Convergence of the ST-Q learning}

We first characterize the convergence performance of the proposed ST-Q learning approach. As shown in Fig. \ref{fig:compareconvergence}, the controller obtained by the ST-Q learning approach eventually converges to the optimal controller obtained by DQG, which clearly outperforms DQP. Note that all three approaches converge quickly. Moreover, Fig. \ref{fig:rlscompare} further demonstrates the convergence performance of the local controller $\hat{K}_i$ at each agent $i$ compared with DQG, i.e., each agent in the ST-Q learning almost has the same convergence behaviour as in DQG. We also evaluate the gap between the controller $\hat{K}$ obtained by the ST-Q learning and the controller $K$ obtained by DQG in the policy iteration, compared with that for the controller $K_{\textrm{partial}}$ obtained by DQP. It can be seen from  Fig.~\ref{fig:differentN} that the gap  $\|\hat{K}-K\|$ quickly converges to $0$, while there exists a significant gap between $K_{\textrm{partial}}$ and $K$ (dashed line). These results together indicate that the proposed ST-Q learning approach can achieve comparable performance with DQG, corroborating the benefits by using state tracking to facilitate an accurate global state estimation in distributed Q-learning.
\begin{figure}
	\begin{subfigure}[b]{0.42\columnwidth}
		\centering
		\caption{Decaying excitation noise.}
		\includegraphics[width=\textwidth]{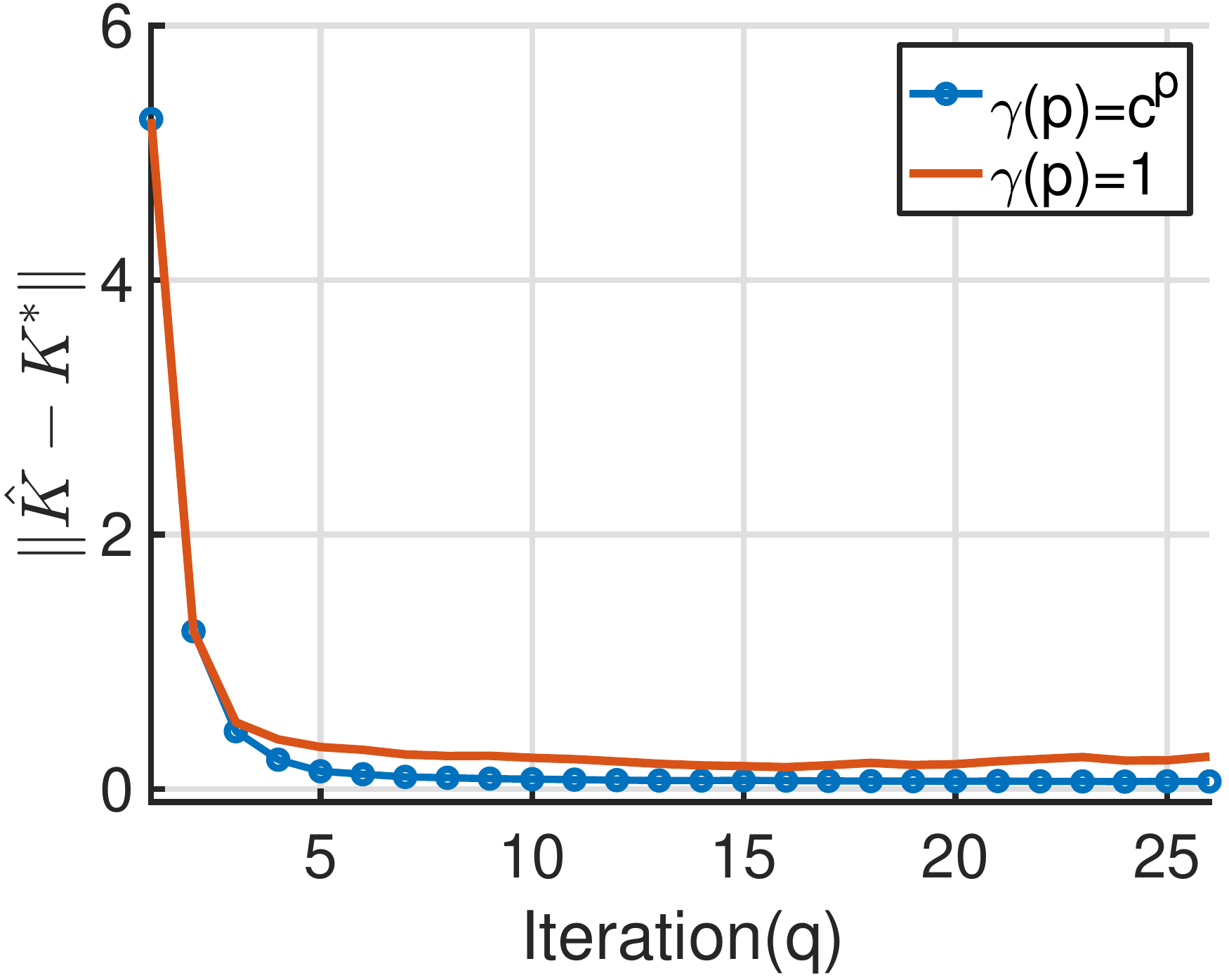}
		\label{fig:noisecompare1}
	\end{subfigure}     
	\hfill
	\begin{subfigure}[b]{0.42\columnwidth}
		\centering
		\caption{Convergence of $\hat{K}-K$.}
		\includegraphics[width=\textwidth]{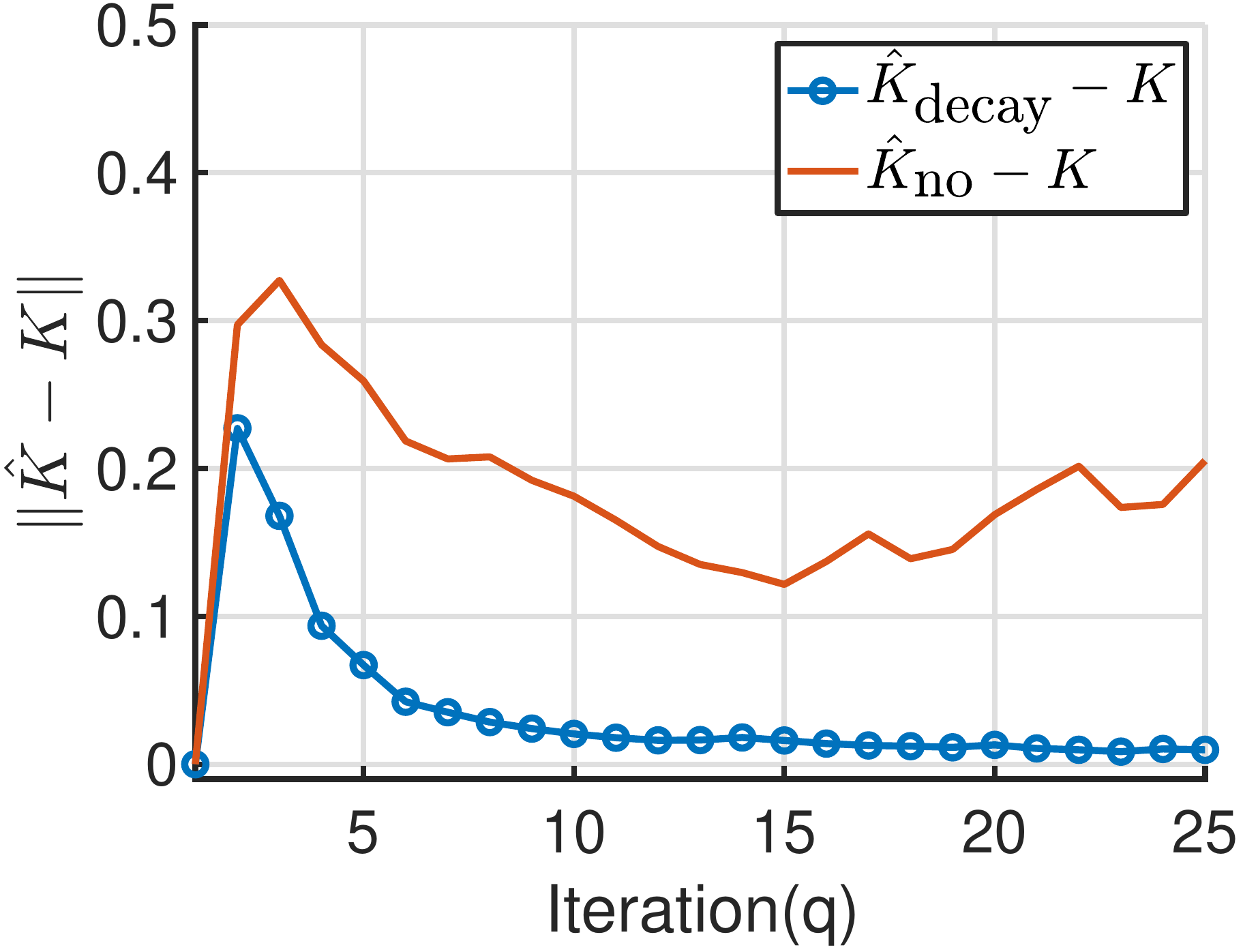}
		\label{fig:noisecompare2}
	\end{subfigure}    
	\caption{Comparison between two types of excitation noise in the ST based  Q-learning. $\hat{K}_{\textrm{decay}}$ is the controller obtained with decaying factor and $\hat{K}_{\textrm{no}}$ is with no decaying factor.}
\end{figure}
\begin{figure}
	\centering
	\begin{subfigure}[b]{0.42\columnwidth}
		\centering
		\caption{ST-Q with $N=50$.}
		\includegraphics[width=\textwidth]{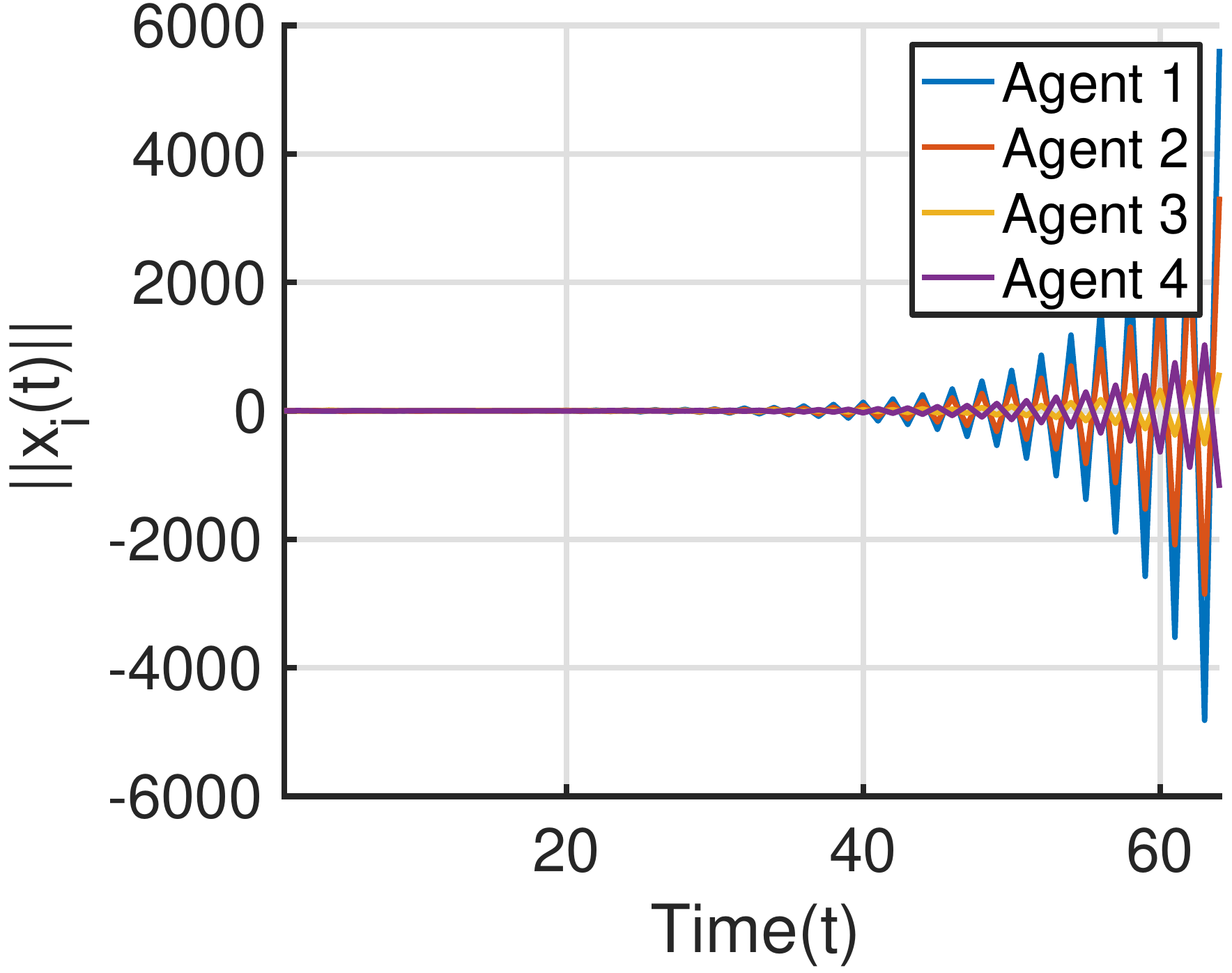}
		\label{fig:smallN}
	\end{subfigure}  
	\hfill
	\begin{subfigure}[b]{0.42\columnwidth}
		\centering
		\caption{ST-Q with different $N$.}
		\includegraphics[width=\textwidth]{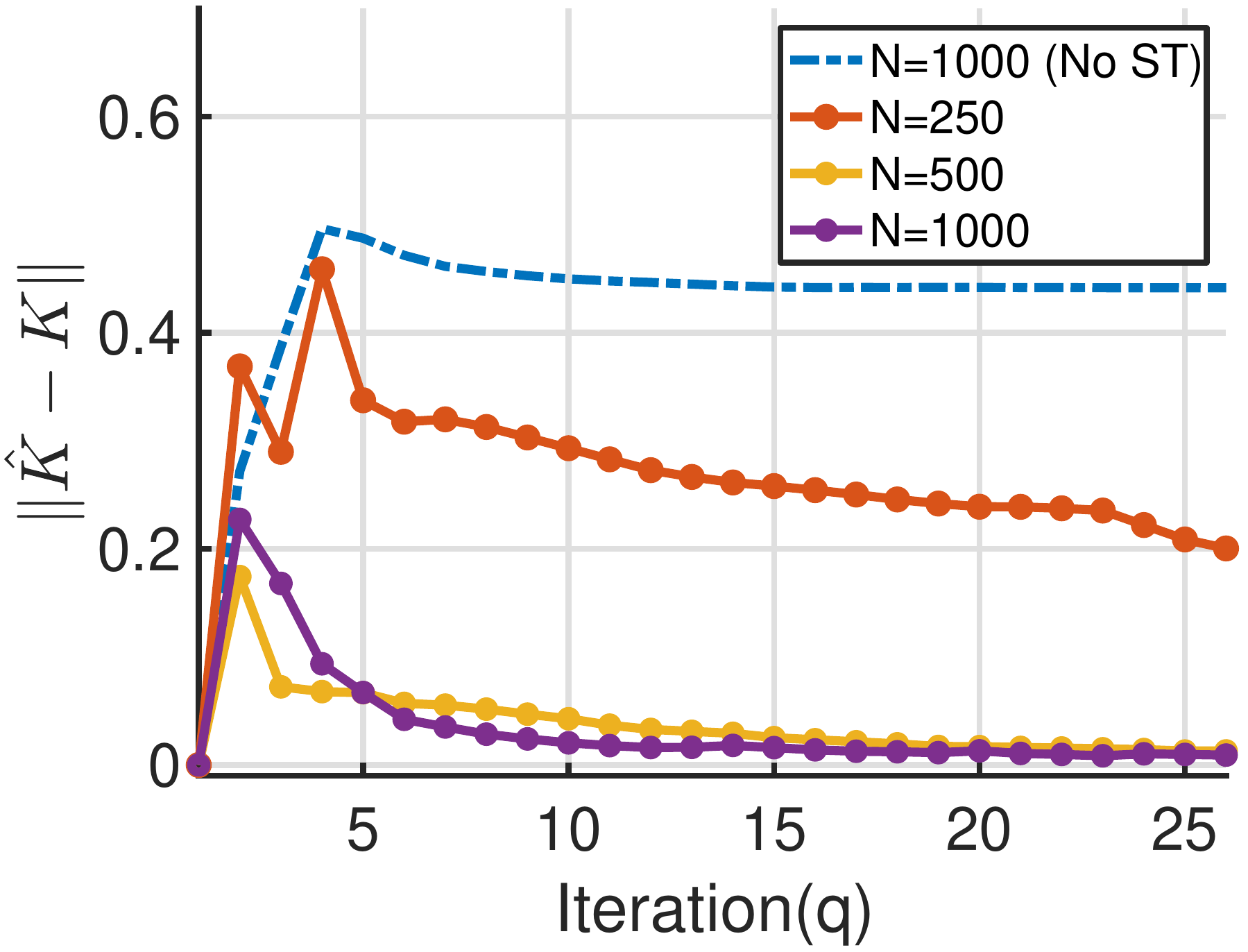}
		\label{fig:differentN}
	\end{subfigure}  
	\caption{Comparisons of different $N$ in the ST-Q learning algorithms.}
	\label{fig:N}
\end{figure}
\subsection{Impact of Hyper-Parameters}
We first evaluate the impact of the step size $\alpha$ on the convergence of the controller $\hat{K}_i$ obtained by the ST-Q learning. In contrast to the step size $0.01$ used in Fig \ref{fig:convergence}, the divergence of the local controller obtained by the ST-Q learning may occur with  a  larger step size (Fig. \ref{fig:largeA}), while a smaller step size may result in slower convergence rate (Fig. \ref{fig:smallA}).  

To examine the impact of the excitation noise, we compare the controller convergence performance under two different cases: (i) the noise is decaying as in Assumption \ref{asu:decaynoise}, and (ii) the noise is not decaying. As demonstrated in Fig. \ref{fig:noisecompare1} and Fig. \ref{fig:noisecompare2}, the controller $\hat{K}$ obtained by the ST-Q learning may not converge when the excitation noise is not decaying, verifying the necessity of Assumption \ref{asu:decaynoise} to guarantee the convergence of the proposed ST-Q learning approach.

Clearly, the performance of the ST-Q learning depends on the estimation accuracy of $\hat{\theta}_{iq}$ in the policy evaluation step, which is directly affected by the value of $N$. Intuitively, as $N$ increases, the estimation accuracy of $\hat{\theta}_{iq}$ improves accordingly, leading to a better performance of the ST-Q learning. Fig. \ref{fig:N} illustrates the impact of $N$ on the  performance of the ST-Q learning. As expected, when $N$ is not large enough, the estimated controller may destabilize the system as shown in Fig. \ref{fig:smallN} due to the lack of  adequate samples needed for achieving a better $\hat{\theta}_{iq}$. And Fig. \ref{fig:differentN}  indicates that the larger $N$ is, the better the performance of the ST-Q learning is. When $N$ is large enough, the statement in Lemma \ref{lemma:estimationerror} where the difference of the estimate $\hat{\theta}_{iq}$ obtained by the ST-based method and the estimate $\theta_{iq}$ obtained by the full observation method is decreasing along with the policy update, is verified in  Fig.~\ref{fig:differentN}.

\section{Conclusions and Future Work}
This work investigates a distributed multi-agent LQR control setup in a networked environment, in which the system dynamics, including the dynamics coupling graph is unknown. Each agent makes individual decisions based on its local observation and messages passed by its neighbors over the communication graph. Within this setting, we propose a multi-agent State Tracking based Q-learning method. Further, the asymptotic analysis on the convergence of the proposed algorithm is provided under mild assumptions. Empirically, in evaluation on an interconnected  system, we demonstrate that the proposed ST-Q learning method outperforms the classic Q-learning with only the partial observation and yields the same optimal controller as the full observation setting. In future work, we shall consider more complicated  communication settings, e.g., (i) communication delay in the network; (ii) time-varying graph. Moreover, it is also of interest to quantify the sampling complexity and the convergence rate of the proposed algorithm.

\bibliographystyle{unsrt}  
\bibliography{references}  

\newpage
\appendix
{\centering{\Large\bf{Appendix}}}
\asucontrollable*
\asuconnect*
\asuweightmatrix*
\asudecaynoise*
\asulr*

\section{Quadratic Structure of the Q-function}\label{app:qfunction}
Substituting the system dynamics \eqref{eqn:subsystem} into the definition of the Q-factor \eqref{eqn:Qfunction}, we can obtain the quadratic structure of the Q-factor. Given a full observation of the global state, we have that
\begin{equation*}
\begin{aligned}
Q_i(x_i(t),u_i(t)) 
&= g_i(x_i(t),u_i(t)) + Q_i(x_i(t),u_i(t+1))\\
&= x_i(t)^{\top} {P_i} x_i(t) + u_i(t)^{\top}R_iu_i(t)+ x_i(t+1)^{\top} S_i x_i(t+1)\\
&=\left(\begin{array}{l}{X}(t) \\ {u_i}(t)\end{array}\right)^{\top}\left(\begin{array}{ccc} {\mathcal{A}_i}^{\top} {S_i \mathcal{A}_i}+\Tilde{P}_i &  {\mathcal{A}_i}^{\top} {S_i B_i}    \\  {B_i}^{\top} {S_i \mathcal{A}_i} &  {B_i}^{\top} {S_i B_i} + R_i \end{array}\right)\left(\begin{array}{l}{X}(t) \\ {u_i}(t)\end{array}\right)\\
& = \left(\begin{array}{l}{X}(t) \\ {u_i}(t)\end{array}\right)^{\top}\left(\begin{array}{ccc} H_{11,i}  &  H_{12,i}   \\ H_{21,i} & H_{22,i}  \end{array}\right)\left(\begin{array}{l}{X}(t) \\ {u_i}(t)\end{array}\right)\\
&= \left(\begin{array}{l}{X}(t) \\ {u_i}(t)\end{array}\right)^{\top}H_i\left(\begin{array}{l}{X}(t) \\ {u_i}(t)\end{array}\right)\\
& \triangleq {y_i}(t)^{\top} {\theta_i}, \\
\Tilde{P}_i &= \begin{bmatrix} 
0_{n \times n} & \dots & 0 \\
\vdots & P_i & \vdots \\
0 &    \dots    & 0 
\end{bmatrix} \in \mathbb{R}^{Ln \times Ln},
\end{aligned}
\end{equation*}
where $\Tilde{P}_i$ is a diagonal matrix with the  $(i,i)$-th block set to be $P_i$. ${y_i}(t)=[{x}_{1}^2(t),{x}_{1}(t){x}_{2}(t), \cdots, {x}_{L}(t)u_i(t),u_i^2(t)]$ is a vector consisting of  all  the quadratic basis over the elements in $[X(t);u_i(t)]$.
Since $H_i$ is symmetric, it suffices to use ${\theta_i} \in \mathbb{R}^{(Ln+m)(Ln+m+1)/2}$ to represent the unknown parameters, i.e., the elements of ${\theta_i}$ are the upper right triangle of $H$ in the correct order. Moreover, ${\mathcal{A}_i}$ is a row vector which stacks subvectors $\{A_{ij}\}_{j\in[L]}$.

\section{Proof of Lemma 1}\label{app:A}
\lemmafullobservation*
\begin{proof}
	The proof of this lemma includes two steps. First, we demonstrate that by replacing recursive least squares (RLS) with stochastic gradient descent (SGD) in the adaptive policy iteration algorithm proposed in \cite{bradtke1994adaptive}, the policy iteration algorithm also generates a sequence of stabilizing controls converging to the optimal in the single agent case. 
	
	Observe that in \cite{bradtke1994adaptive}, only the intermediate result Lemma 2 requires the property of the RLS estimation. Thus we need to prove that the SGD estimation also has the same property. Recall Lemma 2 in \cite{bradtke1994adaptive},
	\begin{lemma*}[Lemma 2  \cite{bradtke1994adaptive}]
		If $\phi_i(t)$ is persistently excited and $N > N_0$, then we have
		\begin{equation*}
		\begin{aligned}
		\|\hat{\theta}_{iq} -\theta_{iq}\|&\leq \varepsilon_{N}\big(\|\hat{\theta}_{i(q-1)} - \theta_{i(q-1)} \| + \|\theta_{i(q-1)}- \theta_{iq})\| \big),
		\end{aligned}
		\end{equation*}
		where $\lim_{N \to \infty}\varepsilon_{N}=0$.
	\end{lemma*}
	This lemma still holds when replacing RLS with SGD. Consider the $q$-th policy iteration where $\theta_{iq}$ is the true parameter vector for the Q-factor with control policy $K_i$.  $\hat{\theta}_{iq}= \hat{\theta}_{iq}(N)$ is the estimate of ${\theta}_{iq}$ at the end of the $q$-th policy iteration. The initial estimate is the final value from the previous policy iteration,  i.e., $\hat{\theta}_{iq}(0)=\hat{\theta}_{i(q-1)}(N)$. Recall the SGD algorithm, 
	
	\begin{equation*}
	\begin{aligned}
	\textstyle \hat{\theta}_{iq}(N) &= \hat{\theta}_{iq}(N-1) - \alpha\phi_i(t_q+N)\cdot\Big(\hat{\theta}_{iq}(N-1)^{\top}\phi_i(t_q+N) - g_i(t_q+N) \Big),\\
	\hat{\theta}_{iq}(N) -\theta_{iq} &= \big(I - \alpha\phi_i(t_q+N)\phi_i^{\top}(t_q+N)\big)(\hat{\theta}_{iq}(N-1) - \theta_{iq}) - \underbrace{\alpha\phi_i(t_q+N)\cdot\Big(\phi_i^{\top}(t_q+N)\theta_{iq} - g_i(t_q+N) \Big)}_{=0} \\
	&= \big(I - \alpha\phi_i(t_q+N)\phi_i^{\top}(t_q+N)\big)(\hat{\theta}_{iq}(N-1) - \theta_{iq})\\
	& = \cdots\\
	&= \prod_{\tau=t_q+1}^{t_q+N}\big(I - \alpha\phi_i(\tau)\phi_i^{\top}(\tau)\big)(\hat{\theta}_{i(q-1)} - \theta_{iq})\\
	&= \prod_{\tau=t_q+1}^{t_q+N}\big(I - \alpha\phi_i(\tau)\phi_i^{\top}(\tau)\big)(\hat{\theta}_{i(q-1)} - \theta_{i(q-1)} + \theta_{i(q-1)}- \theta_{iq}),\\
	\| \hat{\theta}_{iq}(N) -\theta_{iq}\|&\leq \|\prod_{\tau=t_q+1}^{t_q+N}\big(I - \alpha\phi_i(\tau)\phi_i^{\top}(\tau)\big)\| \cdot \big(\|\hat{\theta}_{i(q-1)} - \theta_{i(q-1)} \| + \|\theta_{i(q-1)}- \theta_{iq})\| \big)\\
	&\triangleq \varepsilon_{N}\big(\|\hat{\theta}_{i(q-1)} - \theta_{i(q-1)} \| + \|\theta_{i(q-1)}- \theta_{iq})\| \big),
	\end{aligned}
	\end{equation*}
	where step size $\alpha$ satisfies Assumption \ref{asu:lr} and $\phi_i(t)$ satisfies Assumption \ref{asu:decaynoise}. Hence we obtain that
	\begin{equation*}
	\lim_{N \to \infty}\varepsilon_{N}=0.
	\end{equation*}
	
	Second, we demonstrate that Lemma \ref{lemma:fullobservation} holds when agents have full observation of the global state. Under this setting, \cite{narayanan2016distributed} considers a policy iteration algorithm with the RLS estimation method and further provides the convergence proof by utilizing the single agent result in \cite{bradtke1994adaptive}. By following the same approach in \cite{narayanan2016distributed} and the results  obtained in the preceding step, we are able to obtain the desired result in Lemma \ref{lemma:fullobservation}. See \cite[Theorem~1]{narayanan2016distributed} for detailed proof. 
	
\end{proof}

\section{Proof of Lemma 2 (a)}\label{app:C}
\lemmaestimationerror*
\begin{proof}
	Before proceeding to the proof of Lemma 2, we first characterize the change of the agents' state  (i.e., $x_i(t+1)-x_i(t)$) during the policy iteration. Consider Agent $i$'s state in the $q$-th policy iteration. Notice that $t= t_q + p$, where $t_q$ is the time index at the start of the $q$-th policy iteration and $p$ counts the number of time steps from the start of the $q$-th policy iteration, such that $p\in [1,N]$ and $t\in [t_q,t_q+N]$. Assume that $X_q = X(t_q)$ is the initial global state for the $q$-th policy iteration. Following the system dynamics defined in \eqref{eqn:globalsystem}, we obtain the global state after $p$ steps of policy evaluation as,
	\begin{equation*}
	\begin{aligned}
	X(t_q) &= X_q,\\
	X(t_q + 1) &= AX_q+B(K_qX_q + E(1)) = (A+BK_q)X_q + BE(1),\\
	X(t_q + p) &= (A+BK_q)^p X_{q} + \sum_{\tau=1}^{p}(A+BK_q)^{p-\tau}BE(\tau).\\
	\end{aligned}
	\end{equation*}
	
	\noindent Suppose Assumption \ref{asu:decaynoise} holds and correspondingly we have
	\begin{equation*}
	\begin{aligned}
	X_{q} & =X(t_{q-1} + N) \\
	&= (A+BK_{q-1})^{N}X_{q-1} + \sum_{\tau=1}^{N}(A+BK_{q-1})^{N-\tau}BE(\tau)\\
	&=(A+BK_{q-1})^{N}X_{q-1} + \sum_{\tau=1}^{N}(A+BK_{q-1})^{N-\tau}B\Upsilon(\tau)\boldsymbol{\beta}(\tau)\\
	&= (A+BK_{q-1})^{N}X_{q-1} + \sum_{\tau=1}^{N}(A+BK_{q-1})^{N-\tau}c^{\tau}B\boldsymbol{\beta}(\tau),\\
	\end{aligned}
	\end{equation*}
	which indicates that $\|X_q\|\to 0$ as $N \to \infty$.
	
	Hence, we  can obtain that
	\begin{equation*}
	\begin{aligned}
	X(t_q + p + 1) -X(t_q + p)  =& ((A+BK_q)^{p+1}-(A+BK_q)^{p})X_q \\
	&+ (A+BK_q)^{p}BE(1) + \sum_{\tau=1}^{p}(A+BK_q)^{p-\tau}B(E(\tau+1) - E(\tau)).
	\end{aligned}
	\end{equation*}  
	
	We further define       
	\begin{equation*}
	\begin{aligned}      
	r & \triangleq rank(X(t_q + p + 1) -X(t_q + p)),\\
	\|(A+BK_q)^k\|_2 &= \sigma_{max}^k(A+BK_q) < 1.\\
	\end{aligned}
	\end{equation*}
	
	Let $\|C\|_F$ denote the Frobenius norm of a matrix $C \in \mathbb{R}^{m \times n}$, i.e., $\|C\|_F = \sqrt{\textstyle\sum_{i=1}^m\sum_{j=1}^n c_{ij}^2} $. Now, we have,
	\begin{equation*}
	\begin{aligned}           
	\| X(t_q + p + 1) -X(t_q + p)\|_F  \leq& \sqrt{r}\| X(t_q + p + 1) -X(t_q + p)\|_2 \\ 
	\leq& \sqrt{r}\|(A+BK_q)^{p+1}-(A+BK_q)^{p}\|_2\|X_q\|_2 \\
	&+ \sqrt{r}\|(A+BK_q)^{p}\|_2\|E(1)\|_2 + \sqrt{r}\|\sum_{\tau=1}^{p}(A+BK_q)^{p-\tau}B(E(\tau+1) - E(\tau))\|_2\\
	\leq & \sqrt{r}\|(A+BK_q)^{p+1}-(A+BK_q)^{p}\|_2\|X_q\|_2 + \sqrt{r}\|(A+BK_q)^{p}\|_2\|E(1)\|_2 \\ 
	&+ \sqrt{r}\sum_{\tau=1}^{p}\Big(\|(A+BK_q)^{p-\tau}\|_2\|B(E(\tau+1) - E(\tau))\|_2 \Big)\\
	\leq& \sqrt{r}\|(A+BK_q)^{p+1}-(A+BK_q)^{p}\|_2\|X_q\|_2 +\sqrt{r}\sigma_{max}^t \Upsilon(1)\|\boldsymbol{\beta}(1)\|_2 \\
	&+ \sqrt{r}\sum_{\tau=1}^{p} (\sigma_{max}^{p-\tau} \|B(c^{\tau+1}\boldsymbol{\beta}(\tau+1) - c^{\tau}\boldsymbol{\beta}(\tau))\|_2 )\\
	\leq& \sqrt{r}\|(A+BK_q)^{p+1}-(A+BK_q)^{p}\|_2\|X_q\|_2 +\sqrt{r}\sigma_{max}^p \Upsilon(1)\|\boldsymbol{\beta}(1)\|_2 \\
	&+ c_{\beta}\sqrt{r}p\cdot \textrm{max}\{c,\sigma_{max}\}^p\\
	\triangleq & \delta_x,
	\end{aligned}
	\end{equation*}
	where $c_{\beta} = \textrm{max}_{\tau}\|B(c\boldsymbol{\beta}(\tau+1) - \boldsymbol{\beta}(\tau))\|$. $c_{\beta}$ is finite since the injected noise satisfies the PE condition in Assumption \ref{asu:decaynoise}. Notice that with a fixed $q$, when $p \to \infty$, $\delta_x \to 0$. Note that $K_q$ is a stable controller, such that $\sigma_{max} < 1$.

	Therefore, there exists $N<\infty$ and we can choose $ N_0 \in (0,N)$ such that when $t \in [t_q+N_0,t_{q+1}]$ ($p$ is large enough), the difference between two adjacent  states of Agent $i$ is bounded from above by some arbitrary small $\delta_x>0$, i.e., $\forall i\in [L]$:
	\begin{equation}
	\|x_{i}(t)-x_{i}(t-1)\| \leq \delta_{x}.
	\label{eqn:slowchange}
	\end{equation}
	
	Now we are ready to prove Lemma \ref{lemma:estimationerror} (a).
	
	Define $\epsilon_{ik}(t) =  \sum_{j \in \mathcal{N}_k} w_{ij} ({x}_{k}(t)-{x}_{k}(t-1)) $ and $ \bar{x}_{av,k}(t) = \frac{1}{L} \sum_{j=1}^{L}\bar{x}_{jk}(t) $ (the average estimation towards Agent $k$ at time instance $t$ ). Further, let $d_k = |\mathcal{N}_k|$ denote the  amount of communication neighbors of Agent $k$ in the communication network. The norm of the difference between the ST-based global state estimates $Z_i(t)$ and the true global state $X(t)$ can be shown as
	
	\begin{equation}
	\begin{aligned}
	\|{Z}_i(t)- {X}(t)\|_{F} &=  \sqrt{\sum_{k =1}^{L} \|x_k(t) - \bar{x}_{ik}(t)\|^2_2 }\\
	& \leq  \sum_{k =1}^{L} \sqrt{\|x_k(t) - \bar{x}_{ik}(t)\|^2 }\\
	& = \sum_{k =1}^{L} {\|x_k(t) - \bar{x}_{av,k}(t) + \bar{x}_{av,k}(t) - \bar{x}_{ik}(t)\| }\\
	& \leq \Big( \underbrace{{  \sum_{k =1}^{L}  \|x_{k}(t) - \bar{x}_{av,k}(t)\| }}_{\text{\circled{1}}} + \underbrace{{  \sum_{k =1}^{L}\|\bar{x}_{av,k}(t)- \bar{x}_{ik}(t)\| }}_{\text{\circled{2}}} \Big).\\
	\end{aligned}
	\label{equ:step1}
	\end{equation}
	First we consider term \text{\circled{1}} in \eqref{equ:step1},
	\begin{equation}
	\begin{aligned}
	{ \sum_{k =1}^{L} \|x_k(t) - \bar{x}_{av,k}(t)\| } &=   { \sum_{k =1}^{L} \|x_k(t) - \frac{1}{L}\sum_{j=1}^{L} \bar{x}_{jk}(t)\| }\\
	& = { \sum_{k =1}^{L} \|x_k(t) - \Big(  \bar{x}_{av,k}(0) + \frac{d_k}{L}{x}_{k}(t) - \frac{d_k}{L} {x}_{k}(0)  \Big)\| }\\
	& \leq L \max_k\{\|x_k(t)\|\}
	\triangleq w(t),
	\end{aligned}
	\label{equ:step1bound1}
	\end{equation}
	where we use the deduction of $\bar{x}_{av,k}(t)$:
	\begin{equation}
	\begin{aligned}
	\bar{x}_{av,k}(t) &=   \frac{1}{L}\sum_{j=1}^{L} \bar{x}_{jk}(t)\\
	& = \frac{1}{L}\sum_{j=1}^{L}\sum_{u=1}^{L} w_{ju}\hat{x}_{uk}(t)\\
	& = \frac{1}{L}\sum_{u=1}^{L}\hat{x}_{uk}(t)\\
	& = \frac{1}{L}(\sum_{u\in \mathcal{N}_k}{x}_{k}(t) + \sum_{u\notin \mathcal{N}_k}\bar{x}_{uk}(t-1))\\
	& = \bar{x}_{av,k}(t-1) + \frac{d_k}{L}{x}_{k}(t) - \frac{d_k}{L} {x}_{k}(t-1) \\
	& = \bar{x}_{av,k}(0) + \frac{d_k}{L}(x_k(t) - x_k(0)).
	\end{aligned}
	\end{equation}
	Now, consider term \text{\circled{2}}: Let Assumptions \ref{asu:connect}, \ref{asu:weightmatrix} hold and rewrite $\bar{x}_{ik}(t)$ into the following format,
	\begin{equation}
	\begin{aligned}
	\bar{x}_{ik}(t) &= \sum_{j=1}^{L} w_{ij} \hat{x}_{jk}(t) \\
	&= \sum_{j \notin \mathcal{N}_k} w_{ij} \hat{x}_{jk}(t) +  \sum_{j \in \mathcal{N}_k} w_{ij} \hat{x}_{jk}(t)\\
	&= \underbrace{\sum_{j \notin \mathcal{N}_k} w_{ij} \bar{x}_{jk}(t-1)}_{\text{Weighted Estimation}} +\underbrace{  \sum_{j \in \mathcal{N}_k} w_{ij} {x}_{k}(t) }_{\text{Weighted True State}} \\
	&= \sum_{j} w_{ij} \bar{x}_{jk}(t-1) + \underbrace{  \sum_{j \in \mathcal{N}_k} w_{ij} ({x}_{k}(t)-{x}_{k}(t-1)) }_{\text{Perturbation }\triangleq \epsilon_{ik}(t)}. \\
	\end{aligned}
	\label{equ:average}
	\end{equation}
	
	Following  the same line as in \cite{nedic2018distributed}, we reformulate \eqref{equ:average} as a perturbed consensus problem:
	\begin{equation}
	\begin{aligned}
	\overline{x}_{ik}(t) &=  \sum_{j} w_{ij} \bar{x}_{jk}(t-1) + \epsilon_{ik}(t),\\
	\epsilon_{ik}(t) &=  \sum_{j \in \mathcal{N}_k} w_{ij} ({x}_{k}(t)-{x}_{k}(t-1)).\\
	\end{aligned}
	\end{equation}
	
	For ease of exposition, we rewrite the evolution of the iterates $\overline{x}_{ik}(t)$ in a matrix form. For any coordinate index $l \in [n]$ (n is the dimension of the state vector), we can have the following for the $l-$th coordinate(denoted by a superscripts):
	$$ \overline{x}_{ik}^{l}(t) = \sum_{j} w_{ij} \bar{x}_{jk}^{l}(t-1) + \epsilon_{ik}^{l}(t), \quad \forall l \in [n]. $$
	
	Define $ \bar{x}_k(t-1) = \begin{bmatrix}
	\bar{x}_{1k}(t-1)\\
	\vdots\\
	\bar{x}_{Lk}(t-1)
	\end{bmatrix} $, 
	$\epsilon_p(t) = \begin{bmatrix}
	\epsilon_{1k}(t)\\
	\vdots\\
	\epsilon_{Lk}(t)
	\end{bmatrix} $.
	
	Next, we stack all of the $l-$th coordinates in a column vector, denoted by $\overline{x}_{k}^{l}(t)$, i.e.,
	\begin{equation*}
	\begin{aligned}
	\overline{x}_{k}^{l}(t) = \begin{bmatrix} \overline{x}_{1k}^{l}(t)\\
	\overline{x}_{2k}^{l}(t)\\
	\vdots\\
	\overline{x}_{Lk}^{l}(t)\\
	\end{bmatrix}
	= W \bar{x}_k^{l}(t-1) + \epsilon_{k}^{l}(t).
	\end{aligned}
	\end{equation*}
	
	Moreover, by stacking the column vectors $\overline{x}_{k}^{l}(t),~l \in [n] $ into a matrix $\mathbf{\overline{x}}_{k}(t)$, we further build up the perturbation matrix $\mathbf{{e}}_{k}(t)$ from ${{\epsilon}}_{k}^{l}(t),~l \in [n]$
	\begin{equation}
	\begin{aligned}
	\mathbf{\overline{x}}_{k}(t) = \begin{bmatrix}
	\overline{x}_{k}^{1}(t) & \overline{x}_{k}^{2}(t) & \cdots &
	\overline{x}_{k}^{n}(t)
	\end{bmatrix}
	&= W \mathbf{\overline{x}}_{k}(t-1) + \mathbf{{e}}_{k}(t) \quad \forall t \geq 0.
	\end{aligned}
	\label{equ:stack}
	\end{equation}
	
	Using the recursion, from Eqn.~(\ref{equ:stack}) we see that, for all $ t_q \leq t \leq t_{q+1}$,
	\begin{equation}
	\begin{aligned}
	\mathbf{\overline{x}}_{k}(t) &= W \mathbf{\overline{x}}_{k}(t-1) + \mathbf{{e}}_{k}(t)\\
	&=W(W \mathbf{\overline{x}}_{k}(t-2) + \mathbf{{e}}_{k}(t-1)) + \mathbf{{e}}_{k}(t)\\
	&=(W)^{2} \mathbf{\overline{x}}_{k}(t-2) + (W)^{1}\mathbf{{e}}_{k}(t-1) + \mathbf{{e}}_{k}(t)\\
	&=\cdots\\
	&= (W)^{p} \mathbf{\overline{x}}_{k}(t_q) + \sum_{\tau=1}^{p-1} (W)^{\tau}\mathbf{{e}}_{k}(t-\tau) + \mathbf{{e}}_{k}(t).
	\end{aligned}
	\label{equ:recursive}
	\end{equation}
	
	By multiplying both sides of \eqref{equ:recursive} with matrix $\frac{1}{L} \mathbf{11}^{\top}$, we have
	\begin{equation*}
	\begin{aligned}
	\frac{1}{L} \mathbf{11}^{\top} \mathbf{\overline{x}}_{k}(t) &= \frac{1}{L} \mathbf{11}^{\top}(W)^{p} \mathbf{\overline{x}}_{k}(t_q) +  \Big(\sum_{\tau=1}^{p-1} \frac{1}{L}\mathbf{11}^{\top}(W)^{\tau}\mathbf{{e}}_{k}(t-\tau))\Big) + \frac{1}{L}\mathbf{11}^{\top}\mathbf{{e}}_{k}(t)\\
	&= \frac{1}{L}\mathbf{11}^{\top} \mathbf{\overline{x}}_{k}(t_q) + \sum_{\tau=1}^{p-1} \frac{1}{L}\mathbf{11}^{\top}\mathbf{{e}}_{k}(t-\tau) + \frac{1}{L}\mathbf{11}^{\top}\mathbf{{e}}_{k}(t).
	\end{aligned}
	\end{equation*}
	
	Now, consider $\mathbf{\overline{x}}_{k}(t)-\frac{1}{L} \mathbf{11}^{\top}\mathbf{\overline{x}}_{k}(t) $
	\begin{equation}
	\begin{aligned}
	\mathbf{\overline{x}}_{k}(t)-\frac{1}{L} \mathbf{11}^{\top} \mathbf{\overline{x}}_{k}(t) =& \Big((W)^{p}-\frac{1}{L}\mathbf{11}^{\top}\Big) \mathbf{\overline{x}}_{k}(t_q) + \sum_{\tau=1}^{p-1}  \Big( (W)^{\tau} - \frac{1}{L}\mathbf{11}^{\top} \Big)\mathbf{{e}}_{k}(t-\tau)\\
	& + \Big( I -\frac{1}{L}\mathbf{11}^{\top} \Big)  \mathbf{{e}}_{k}(t) .
	\end{aligned}
	\label{equ:averagedifference}
	\end{equation}
	
	By taking the F-norm of the both sides of \eqref{equ:averagedifference}, we obtain,
	\begin{equation}
	\begin{aligned}
	\|\mathbf{\overline{x}}_{k}(t)-\frac{1}{L} \mathbf{11}^{\top} \mathbf{\overline{x}}_{k}(t)\|_{F}  \leq & \|\big((W)^{p}-\mathbf{11}^{\top}\big) \mathbf{\overline{x}}_{k}(t_q) \|_{F} +  \sum_{\tau=1}^{p-1}  \|\big( (W)^{\tau} - \frac{1}{L}\mathbf{11}^{\top} \big)\mathbf{{e}}_{k}(t-\tau)\|_{F}\\
	& + \|\big( I -\frac{1}{L}\mathbf{11}^{\top} \big)  \mathbf{{e}}_{k}(t) \|_{F}\\
	\leq& \|\big((W)^{p}-\mathbf{11}^{\top}\big)\|_{F} \|\mathbf{\overline{x}}_{k}(t_q) \|_{F} +  \sum_{\tau=1}^{p-1} \Big( \| (W)^{\tau} - \frac{1}{L}\mathbf{11}^{\top} \|_{F}  \|\mathbf{{e}}_{k}(t-\tau)\|_{F}\Big)\\
	& + \|\big( I -\frac{1}{L}\mathbf{11}^{\top} \big)\|_{F}  \|\mathbf{{e}}_{k}(t) \|_{F}.\\
	\end{aligned}
	\end{equation}
	The following lemma (Lemma 5 \cite{nedic2018distributed}) is required here. 
	\begin{lemma}
		Let the graph $G^c$ satisfy Assumption \ref{asu:connect} and let the weight matrix $W$ satisfy Assumption \ref{asu:weightmatrix}. Then, for all $s \geq 0$, 
		\begin{equation*}
		\Big( [W^{s}]_{ij} - \frac{1}{L} \Big)^2 \leq \Big( 1- \frac{\eta}{2L^2}\Big)^{s-1},\quad \forall i,j \in [L] .
		\end{equation*}
		\label{lemma:weight}
	\end{lemma}
	
	Based Lemma \ref{lemma:weight}, we have
	\begin{equation*}
	\begin{aligned}
	\| (W)^{\tau} - \frac{1}{L}\mathbf{11}^{\top} \|_{F} &=\sqrt{\sum_{i=1}^{L} \sum_{j=1}^{L}\Big( [W^{\tau}]_{ij} - \frac{1}{L}\Big)^2  } \leq L \sqrt{\Big(1-\frac{\eta}{2L^2}\Big)^{\tau-1}}.
	\end{aligned}
	\end{equation*}
	Following the fact that $\sqrt{1-\mu} \leq 1- \frac{\mu}{2},\forall \mu\in(0,1)$, we further have,
	\begin{equation*}
	\begin{aligned}
	\| (W)^{\tau} - \frac{1}{L}\mathbf{11}^{\top} \|_{F} \leq  L {c_w^{\tau-1}},~c_w  = 1- \frac{\eta}{4L^2}.
	\end{aligned}
	\end{equation*}
	For the norm $\|I -\frac{1}{L}\mathbf{11}^{\top}\|_F $, we have
	\begin{equation}
	\|I -\frac{1}{L}\mathbf{11}^{\top}\|_F = \sqrt{L \Big( 1-\frac{1}{L}\Big)^2 + (L-1)L \frac{1}{L^2}} = \sqrt{L-1}.
	\end{equation}
	Now, we obtain that
	\begin{equation}
	\begin{aligned}
	\|\mathbf{\overline{x}}_{k}(t)-\frac{1}{L} \mathbf{11}^{\top} \mathbf{\overline{x}}_{k}(t)\|_{F} & \leq Lc_w^{p-1} \| \mathbf{\overline{x}}_{k}(t_q) \|_{F} + L \Big(\sum_{\tau=1}^{p-1} c_w^{\tau-1} \|\mathbf{\overline{e}}_{k}(t-\tau)\|_F  \Big) + \sqrt{L-1} \|\mathbf{\overline{e}}_{k}(t)\|_F.
	\end{aligned}
	\end{equation}
	This equation is equivalent to
	\begin{equation}
	\begin{aligned}
	\sqrt{\sum_{i=1}^{L} \|\bar{x}_{ik}(t) - \bar{x}_{av,k}(t)\|^2} \leq Lc_w^{p-1} \sqrt{\sum_{i=1}^{L}\|\bar{x}_{ik}(t_q)\|^2} & + L\Big( \textstyle\sum_{\tau=t-p+1}^{t-1} c_w^{t-\tau-1} \sqrt{\sum_{i=1}^{L} \| \epsilon_{ik}(\tau) \|^2} \Big)\\
	& + \sqrt{L-1} \sqrt{ \sum_{i=1}^{L} \| \epsilon_{ik}(t) \|^2}.
	\end{aligned}
	\label{equ:consensus}
	\end{equation}
	
	Recall that $\epsilon_{ik}(t) =  \sum_{j \in \mathcal{N}_k} w_{ij} ({x}_{k}(t)-{x}_{k}(t-1)) $, $ \bar{x}_{av,k}(t) = \frac{1}{L} \sum_{j=1}^{L}\bar{x}_{jk}(t) $ (the average estimation towards Agent $k$ at time instance $t$ ). Now, we are ready to obtain the upper bound on $\sum_{k =1}^{L} \|\bar{x}_{av,k}(t)- \bar{x}_{ik}(t)\|$:
	\begin{equation}
	\begin{aligned}
	\sum_{k =1}^{L} \|\bar{x}_{av,k}(t)- \bar{x}_{ik}(t)\|=&  \sum_{k \notin \mathcal{N}_i} \|\bar{x}_{av,k}(t)- \bar{x}_{ik}(t)\| +  \sum_{k \in \mathcal{N}_i} \|\bar{x}_{av,k}(t)- {x}_{k}(t)\|\\
	\leq&  L^{\frac{1}{2}} \sqrt{\sum_{i=1}^{L} \|\bar{x}_{ik}(t) - \bar{x}_{av,k}(t)\|^2} + w(t)\\
	\leq& L^{\frac{3}{2}}c_w^{p-1} \sqrt{\sum_{i=1}^{L}\|\bar{x}_{ik}(t_q)\|^2}  + L^{\frac{3}{2}}\Big( \textstyle\sum_{\tau=t-p+1}^{t-1} c_w^{t-\tau-1} \sqrt{\sum_{i=1}^{L} \| \epsilon_{ik}(\tau) \|^2} \Big)\\
	& + L^{\frac{1}{2}}\sqrt{L-1} \sqrt{ \sum_{i=1}^{L} \| \epsilon_{ik}(t) \|^2} + (L-d_k)\|x_k(t)\|\\
	\triangleq& v(t),
	\end{aligned}
	\label{equ:step1bound2}
	\end{equation}
	where the first inequality is based on the H\"{o}lder's inequality. Note that $t$ denotes the time instance from the start of the algorithm, i.e., $t = t_q + p$. Thus, we can obtain the following result from \eqref{eqn:slowchange}
	\begin{equation*}
	\begin{aligned}
	\lim_{p\to \infty}\|\epsilon_{ik}(t)\| &=  \lim_{p\to \infty}\|\epsilon(p)\| = 0.
	\end{aligned}
	\end{equation*}
	
	\noindent Besides, the second term of \eqref{equ:consensus} satisfies
	\begin{equation*}
	\begin{aligned}
	&\sum_{\tau=t-p+1}^{t-1} c_w^{t-\tau-1} \sqrt{\sum_{i=1}^{L} \| \epsilon_{ik}(\tau) \|^2} \\
	&\leq  \sum_{\tau=t-p+1}^{t-1} Lc_w^{t-\tau-1}\|   \epsilon_{ik}(\tau)\|\\
	&= \sum_{\tau=t-p+1}^{t-1} Lc_w^{t-\tau-1} \frac{1}{\sum_{\tau=t-p+1}^{t-1} Lc_w^{t-\tau-1}} \sum_{\tau=t-p+1}^{t-1} Lc_w^{t-\tau-1}\|   \epsilon_{ik}(\tau)\|\\
	&=\frac{L(1-c_w^{p-1})}{1-c_w}\Big( \frac{1}{\sum_{\tau=t-p+1}^{t-1} Lc_w^{t-\tau-1}} \sum_{\tau=t-p+1}^{t-1} Lc_w^{t-\tau-1}\|   \epsilon_{ik}(\tau)\|\Big).\\
	\end{aligned}
	\end{equation*}
	Let $p \to \infty$. We have
	\begin{equation}
	\begin{aligned}
	\lim_{p\to \infty} \sum_{\tau=t-p+1}^{t-1} c_w^{t-\tau-1} \sqrt{\sum_{i=1}^{L} \| \epsilon_{ik}(\tau) \|^2} &\leq \frac{L}{1-c_w}\lim_{p \to \infty}\|\epsilon(p)\| = 0,
	\end{aligned}
	\end{equation}
	where the right side follows Mazur's Lemma that any convex combination of a convergent sequence $\{\epsilon(p)\}$ converges to the same limit as the sequence itself.
	
	Following the preceding results, finally we obtain that
	\begin{equation}
	\begin{aligned}
	\|{Z}_i(t)- {X}(t)\| & \leq \Big( w(t) + v(t)\Big)  \triangleq \delta(t),\quad t_q \leq t \leq t_{q+1}.
	\end{aligned}
	\label{equ:step1bound}
	\end{equation}
	Note that when $t \rightarrow \infty$ ($p \to \infty$), $w(t) \rightarrow 0$ (see \eqref{equ:step1bound1}), $v(t) \rightarrow 0$ (see \eqref{equ:step1bound2}). Thus, for all  $\delta>0$, there exist $N$ and  $N_0\in(0,N)$ such that, $\forall t \in [t_q+N_0, t_{q+1}]$,
	\begin{equation}
	\begin{aligned}
	\|{Z}_i(t)- {X}(t)\| & \leq \delta,~\forall i \in [L],
	\end{aligned}
	\label{eqn:accuratesamples}
	\end{equation}
	which indicates that with large enough $N$, we are able to collect $(N-N_0)$ samples with relative error $\delta$ (See Fig. \ref{fig:accuratesamples}).
	\begin{figure}
		\centering
		\includegraphics[width=.53\textwidth]{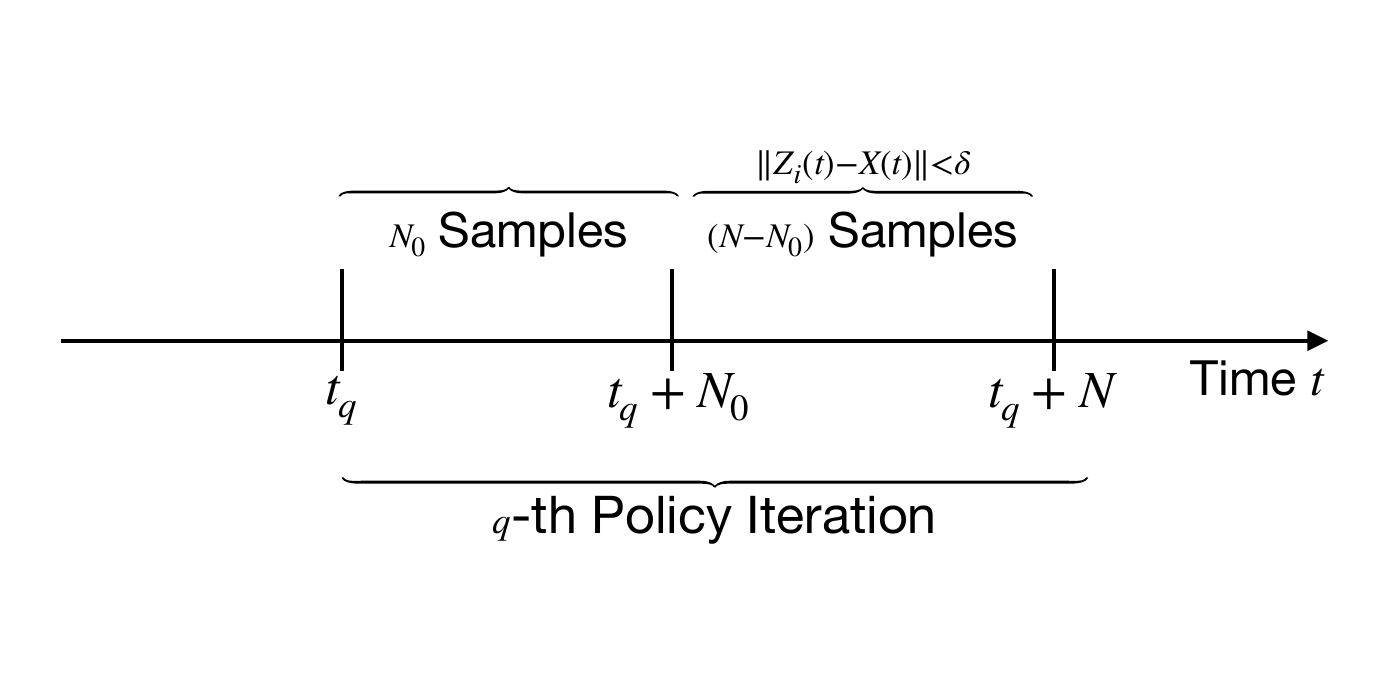}
		\caption{Illustration of the result in \eqref{eqn:accuratesamples}.}
		\label{fig:accuratesamples}
	\end{figure}

\end{proof}
\section{Proof of Lemma 2 (b)}\label{app:D}
\lemmaestimationerror*
\begin{proof}
	Recall the SGD update rule in \eqref{eqn:gradient},
	
	\begin{equation}
	\begin{aligned}
	\textstyle \theta_{iq}(p+1) &= \theta_{iq}(p) - \alpha\phi_i(t)\cdot\Big(\theta_{iq}(p)^{\top}\phi_i(t) - g_i(t) \Big) , \\  
	\textstyle \hat{\theta}_{iq}(p+1) &= \hat{\theta}_{iq}(p) - \alpha\hat{\phi}_i(t)\cdot\Big(\hat{\theta}_{iq}(p)^{\top}\hat{\phi}_i(t) - \hat{g}_i(t) \Big).  \\  
	\end{aligned}
	\end{equation}
	
	We first define
	\begin{equation}
	\begin{aligned}
	\Phi_i(t_q+N-\tau) &\triangleq I-\alpha\phi_i(t_q+N-\tau)\phi_i^{\top}(t_q+N-\tau),\\
	\Pi_i(M) &\triangleq \prod_{m=1}^{M}\Phi_i(t_q+N-m),\\
	G_i(t_q+N-\tau)& \triangleq \alpha\phi_i(t_q+N-\tau) g_i(t_q+N-\tau).\\ 
	\end{aligned}
	\end{equation}
	Next, we use recursion to obtain the relationship between $\theta_{iq} = \theta_{iq}(N)$ and $\theta_{i(q-1)} = \theta_{iq}(1)$,
	\begin{equation}
	\begin{aligned}
	\theta_{iq} &= 
	\Pi_i(N) \theta_{i(q-1)} + \sum_{\tau=2}^{N} \Pi_i(\tau-1)G_i(t_q+N-\tau) + G_i(t_q+N-1),\\
	\hat{\theta}_{iq} &= 
	\hat{\Pi}_i(N) \hat{\theta}_{i(q-1)} + \sum_{\tau=2}^{N}\hat{\Pi}_i(\tau-1) \hat{G}_i(t_q+N-\tau) +  \hat{G}_i(t_q+N-1).\\
	\end{aligned}
	\end{equation}
	\noindent Therefore, we have
	\begin{equation}
	\begin{aligned}
	\|  {\theta}_{iq} - \hat{\theta}_{iq}\| =&  \|  {\Pi}_i(N) \theta_{i(q-1)}-   \hat{\Pi}_i(N)  \hat{\theta}_{i(q-1)}\| \rightarrow \text{\circled{\scriptsize1}} \\
	&+\| G_i(t_q+N-1)-\hat{G}_i(t_q+N-1) \| \rightarrow \text{\circled{\scriptsize2}}\\
	&+\|  \sum_{\tau=2}^{N}\Big( \Pi_i(\tau-1)G_i(t_q+N-\tau)  -   \hat{\Pi}_i(\tau-1) \hat{G}_i(t_q+N-\tau)\Big)\| \rightarrow \text{\circled{\scriptsize3}}.\\
	\end{aligned}
	\end{equation}
	
	In order to utilize the result in Appendix \ref{app:C}, we first explore the relationship between $\phi_i(t)$ and $\|Z_i(t) - X(t) + \hat{u}_i(t) - u_i(t)\|$. It can be shown that the norm of $y_i(t)$ is equivalent to the $F-$norm of the  product of two matrices, i.e. $E_i(t)M_i(t)$,
	\begin{equation}
	\begin{aligned}
	y_i(t) &= \begin{bmatrix}
	x_1^2(t) & x_1(t)x_2(t) & x_1(t)x_3(t) \cdots u_i^2(t)
	\end{bmatrix} \\
	&= 
	\underbrace{ \begin{bmatrix}
		x_1(t) &  & & & \\
		&  x_2(t) &  & & \\
		&   &  \ddots & & \\
		&   & & x_L(t) & \\
		&   &  & & u_i(t)\\
		\end{bmatrix} }_{\triangleq E_i(t)}
	\underbrace{\begin{bmatrix}
		x_1(t) & x_2(t) & \cdots & x_L(t) & u_i(t)\\
		x_2(t) & x_3(t) & \cdots &  u_i(t) &\\
		\vdots& \vdots & & &\\
		u_i(t) & & & & \\
		\end{bmatrix}.}_{\triangleq  M_i(t)}
	\end{aligned}
	\end{equation}
	
	Then, we can apply the result in Appendix \ref{app:C}  ($\delta(t)$ is  defined in \eqref{equ:step1bound}, $\hat{K}_{i(q-1)}$ denotes the controller obtained by using estimated global state, $(q-1)$ denotes the controller is updated in the last policy update). Note that $ t \in [t_q,t_{q+1}]$ denotes the $q-$th policy evaluation time instant.
	\begin{equation}
	\begin{aligned}
	T_{q-1} &\triangleq \|\theta_{i(q-1)} - \hat{\theta}_{i(q-1)} \|, \\
	\|u_i(t) - \hat{u}_i(t)\|&=\|\hat{K}_{i(q)}Z_i(t) - K_{i(q)}X(t)\|,\\
	&\leq \|\hat{K}_{i(q)}-{K}_{i(q)}\| \|X(t)\| + \|X(t)-Z_i(t)\|\|\hat{K}_{i(q)}\|,\\
	&\leq \kappa \|X(t)\| T_{q-1} + b_k\delta(t),
	\end{aligned}
	\end{equation}
	where $b_k < \infty$ is the upper bound of $\|\hat{K}_{i(q-1)}\|$. The last inequality follows that
	\begin{equation*}
	\begin{aligned}
	\hat{K}_{i(q)} &= - \hat{H}_{i(q-1),22}^{-1} \hat{H}_{i(q-1),21},\\
	{K}_{i(q)} &= -{H}_{i(q-1),22}^{-1} {H}_{i(q-1),21}.\\ 
	\end{aligned}
	\end{equation*}
	Then, we have
	\begin{equation*}
	\begin{aligned}
	\hat{K}_{i(q)} -   {K}_{i(q)} &= - \hat{H}_{i(q-1),22}^{-1} \hat{H}_{i(q-1),21} + {H}_{i(q-1),22}^{-1} {H}_{i(q-1),21}\\
	& = {H}_{i(q-1),22}^{-1}(({H}_{i(q-1),21} -\hat{H}_{i(q-1),21} ) + (\hat{H}_{i(q-1),22}-{H}_{i(q-1),22})\hat{H}_{i(q-1),22}^{-1}\hat{H}_{i(q-1),21}),
	\end{aligned}
	\end{equation*}
	\begin{equation*}
	\begin{aligned}
	\|\hat{H}_{i(q-1),22}-{H}_{i(q-1),22}\| & \leq \|\hat{\theta}_{i(q-1)} - {\theta}_{i(q-1)}\|,\\
	\|{H}_{i(q-1),22}\| & \leq \|\theta_{i(q-1)}\| .\\ 
	\end{aligned}
	\end{equation*}
	Since the estimated parameters are bounded and $\kappa > 0$ is a finite  constant, it follows that
	\begin{equation*}
	\begin{aligned}
	\| \hat{K}_{i(q)} -   {K}_{i(q)} \| \leq \kappa  \|\hat{\theta}_{i(q-1)} - {\theta}_{i(q-1)}\|.
	\end{aligned}
	\end{equation*}
	Therefore,
	\begin{equation}
	\begin{aligned}
	\|g_i(t) - \hat{g}_i(t)\|&=\|\langle x_i(t),x_i(t)\rangle_P +\langle u_i(t),u_i(t)\rangle_R - \langle \hat{x}_i(t),\hat{x}_i(t)\rangle_P +\langle \hat{u}_i(t),\hat{u}_i(t)\rangle_R\|\\
	=&\|\langle x_i(t)-\hat{x}_i(t),x_i(t)\rangle_P + \langle \hat{x}_i(t),x_i(t)-\hat{x}_i(t)\rangle_P \\
	&+\langle u_i(t)-\hat{u}_i(t),u_i(t)\rangle_R - \langle \hat{u}_i(t),u_i(t)-\hat{u}_i(t)\rangle_R\|\\
	\leq& \lambda_{max}(P)(2b_x\delta(t)) + 2\lambda_{max}(R)(\kappa \|X(t)\| T_{q-1} + b_k\delta(t))b_u\\
	=& c_1\delta(t) + c_2  \|X(t)\| T_{q-1},\\
	c_1 \triangleq&  2b_x\lambda_{max}(P)+b_kb_u\lambda_{max}(R)<\infty,\\
	c_2 \triangleq& \lambda_{max}(R)\kappa<\infty,\\
	\langle x_i(t),x_i(t) \rangle_{P}  =& x_i(t)^{T}Px_i(t),\\
	\end{aligned}
	\label{equ:g}
	\end{equation}
	where $b_x < \infty$ is the upper bound of the system state $\|x_i(t)\|$ and $b_u < \infty$ is the upper bound of the system input $\|u_i(t)\|$. $\lambda_{max}(P)$ is the largest eigenvalue of matrix $P$ and $\lambda_{max}(R)$ is the largest eigenvalue of matrix $R$. Hence, we obtain the following inequality with regard to $E_i$ and $M_i$
	\begin{equation}
	\begin{aligned}
	\|E_i(t) - \hat{E}_i(t) \| &= \|Z_i(t) - X(t) \| + \|u_i(t) - \hat{u}_i(t)\| \\
	& \leq 2b_k \delta(t) + \kappa \|X(t)\| T_{q-1} \\
	&\triangleq c_3 \delta(t) + c_4 \|X(t)\|T_{q-1},\\
	\|M_i(t) - \hat{M}_i(t) \| &\leq L \|Z_i(t) - X(t) \| + (L+1) \|u_i(t) - \hat{u}_i(t)\| \\
	& \leq (2L+1) b_k \delta(t) + (L+1)\kappa \|X(t)\| T_{q-1} \\
	&\triangleq c_5 \delta(t) + c_6 \|X(t)\|T_{q-1}.
	\end{aligned}
	\end{equation}
	Furthermore, we can obtain that
	\begin{equation}
	\begin{aligned}
	\|\phi_i(t) - \hat{\phi}_i(t)\| &= \| E_i(t)M_i(t) - E_i(t+1)M_i(t+1) - \hat{E}_i(t)\hat{M}_i(t) + \hat{E}_i(t+1)\hat{M}_i(t+1) \|\\
	&\leq \| E_i(t)M_i(t) -\hat{E}_i(t)\hat{M}_i(t)\| + \|E_i(t+1)M_i(t+1) -\hat{E}_i(t+1)\hat{M}_i(t+1)  \|\\
	& \leq (c_3b_m+c_5b_x)(\delta(t) +\delta(t+1)) + (c_4b_m+c_6b_x)(\|X(t)\|+\|X(t+1)\|))T_{q-1}\\
	& \triangleq c_{7} (\delta(t) +\delta(t+1)) + c_8(\|X(t)\|+\|X(t+1)\|) T_{q-1}.
	\end{aligned}
	\label{equ:phi}
	\end{equation}
	Thus,
	\begin{equation}
	\begin{aligned}
	\|(I -\alpha \hat{\phi}_i(t)\hat{\phi}_i^{\top}(t) ) - (I -\alpha {\phi}_i(t){\phi}_i^{\top}(t) ) \|
	&= \|\alpha \hat{\phi}_i(t)\hat{\phi}_i^{\top}(t)-\alpha {\phi}_i(t){\phi}_i^{\top}(t)\|\\
	&\leq c_{9} (\delta(t) +\delta(t+1)) + c_{10}(\|X(t)\|+\|X(t+1)\|) T_{q-1}.
	\end{aligned}
	\label{equ:phiphi}
	\end{equation}
	Now, we are ready to analyze term \circled{\scriptsize1}: $\|  {\Pi}_i(N) \theta_{i(q-1)}-   \hat{\Pi}_i(N)  \hat{\theta}_{i(q-1)}\|$. 
	Note that
	\begin{equation}
	\begin{aligned}
	\|{\Pi}(n)\| \leq {c}_{\pi}^{n} \leq \bar{c}_{\pi}^n,\\
	\|\hat{\Pi}(n)\| \leq \hat{c}_{\pi}^{n} \leq \bar{c}_{\pi}^n,\\
	\bar{c}_{\pi} = \max\{c_{\pi},\hat{c}_{\pi}\},
	\end{aligned}
	\label{equ:pipi0}
	\end{equation}
	where $c_{\pi}=\max_t\{\|\Phi_i(t)\|\}<1 $ and $\hat{c}_{\pi}=\max_t\{\|\hat{\Phi}_i(t)\|\}<1 $.
	
	\noindent Hence, we have  
	\begin{equation}
	\begin{aligned}
	\|  {\Pi}_i(N) -\hat{\Pi}_i(N)\|
	&=\| \prod_{m=1}^{N}\Phi_i(t_q+N-m) - \prod_{m=1}^{N}\hat{\Phi}_i(t_q+N-m)\|\leq \bar{c}_{\pi}^N.
	\end{aligned}
	\label{equ:pipi}
	\end{equation}
	Now, we consider
	\begin{equation}
	\begin{aligned}
	\|  {\Pi}(N) \theta_{i(q-1)}-  \hat{\Pi}(N)  \hat{\theta}_{i(q-1)}\| &\leq \|\Pi(N) - \hat{\Pi}(N)\|\|\hat{\theta}_{i(q-1)}\| + \|\hat{\theta}_{i(q-1)}- {\theta}_{i(q-1)}\|\| {\Pi}(N)\|\\
	&\leq c_{11}\bar{c}_{\pi}^N + \|\Pi(N)\| T_{q-1}\\
	& \triangleq \xi_1.
	\end{aligned}
	\label{equ:item1}
	\end{equation}
	Notice that as $N \to \infty$, $\xi_1 \to 0$.
	
	\noindent Similarly, we analyze term \circled{\scriptsize2}: $\| G(t_q+N-1)-\hat{G}(t_q+N-1) \|$,
	\begin{equation}
	\begin{aligned}
	\| G(t_q+N-1)-\hat{G}(t_q+N-1) \| =&\alpha\|\phi^{\top}(t_q+N-1)g(t_q+N-1)-\hat{\phi}^{\top}(t_q+N-1)\hat{g}(t_q+N-1)\|\\
	\leq& \| \phi^{\top}(t_q+N-1)-\hat{\phi}^{\top}(t_q+N-1)\| \|g(t_q+N-1) \| \\
	&+ \|g(t_q+N-1)-\hat{g}(t_q+N-1)\| \|\hat{\phi}^{\top}(t_q+N-1) \|\\
	\leq & c_{12}\delta(t_q+N) + c_{13}\delta(t_q+N-1) \\
	&+ \big(c_{14}\|X(t_q+N)\| + c_{15}\|X(t_q+N-1)\|\big)T_{q-1}\\
	\triangleq& \xi_2,
	\end{aligned}
	\label{equ:gg}
	\end{equation}
	where the second inequality follows  from \eqref{equ:g} and \eqref{equ:phi}. Notice that as $N \to \infty$, $\xi_2 \to 0$. \eqref{equ:gg} also indicates that for all $\epsilon_{2}>0$, there exists a $0<N_2<\infty$, when $N > N_2$, such that,
	\begin{equation}
	\begin{aligned}
	\| G(t_q+N)-\hat{G}(t_q+N) \| &<  \epsilon_{2},\\
	\| G(t_q+N)\| &< \epsilon_{2}.
	\end{aligned}
	\label{equ:boundG}
	\end{equation}
	
	\noindent Consider term \circled{\scriptsize3}: $\|  \sum_{\tau=2}^{N}\Big( \Pi(\tau-1)G(t_q+N-\tau)  -   \hat{\Pi}(\tau-1) \hat{G}(t_q+N-\tau)\Big)\|$,
	
	\begin{equation}
	\begin{aligned}
	&\|  \sum_{\tau=2}^{N}\Big( \Pi(\tau-1)G(t_q+N-\tau)  -   \hat{\Pi}(\tau-1) \hat{G}(t_q+N-\tau)\Big)\|\\
	\leq&  \sum_{\tau=2}^{N}\|\Pi(\tau-1)-\hat{\Pi}(\tau-1)\| \|G(t_q+N-\tau) \| + \|G(t_q+N-\tau)-\hat{G}(t_q+N-\tau) \| \| \hat{\Pi}(\tau-1)\|\\
	=&  \sum_{\tau=2}^{N^{'}}\Big( \|\Pi(\tau-1)-\hat{\Pi}(\tau-1)\| \|G(t_q+N-\tau) \| + \|G(t_q+N-\tau)-\hat{G}(t_q+N-\tau) \| \| \hat{\Pi}(\tau-1)\|\Big)\\
	& + \sum_{\tau=N{'}}^{N}\Big( \|\Pi(\tau-1)-\hat{\Pi}(\tau-1)\| \|G(t_q+N-\tau) \| + \|G(t_q+N-\tau)-\hat{G}(t_q+N-\tau) \| \| \hat{\Pi}(\tau-1)\|\Big),\\
	\end{aligned}
	\label{equ:thetatheta}
	\end{equation}
	where $N^{'} = N-N_2-1 > N_2$ (assuming $N$ is large enough). Using the result from \eqref{equ:pipi0} and \eqref{equ:boundG},  we further obtain,
	\begin{equation}
	\begin{aligned}
	&\|\sum_{\tau=2}^{N}\Big( \Pi(\tau-1)G(t_q+N-\tau)  - \hat{\Pi}(\tau-1) \hat{G}(t_q+N-\tau)\Big)\|\\
	\leq &  \sum_{\tau=2}^{N^{'}}\Big( \|\Pi(\tau-1)-\hat{\Pi}(\tau-1)\| \|G(t_q+N-\tau) \| + \|G(t_q+N-\tau)-\hat{G}(t_q+N-\tau) \| \| \hat{\Pi}(\tau-1)\|\Big)\\
	& +  \sum_{\tau=N{'}}^{N} \Big( \|\Pi(\tau-1)-\hat{\Pi}(\tau-1)\| \|G(t_q+N-\tau) \| + \|G(t_q+N-\tau)-\hat{G}(t_q+N-\tau) \| \| \hat{\Pi}(\tau-1)\|\Big)\\
	< & \sum_{\tau=2}^{N^{'}}\Big( \bar{c}_{\pi}^{\tau-1}\epsilon_2 \Big) + \sum_{\tau=N{'}}^{N}\Big(b_g\bar{c}_{\pi}^{\tau-1} \Big) \\
	= & \frac{\bar{c}_{\pi}(1-\bar{c}_{\pi}^{N^{'}-1})}{1-\bar{c}_{\pi}} \epsilon_2  + \frac{\bar{c}_{\pi}^{N^{'}-1}(1-\bar{c}_{\pi}^{N-N^{'}+1})}{1-\bar{c}_{\pi}}b_g\\
	= &  \frac{\bar{c}_{\pi}(1-\bar{c}_{\pi}^{N-N_2-2})}{1-\bar{c}_{\pi}} \epsilon_2  + \frac{\bar{c}_{\pi}^{N}(\bar{c}_{\pi}^{-2-N_2}-1)}{1-\bar{c}_{\pi}}b_g\\
	\triangleq & \xi_3,
	\end{aligned}
	\label{equ:thetatheta2}
	\end{equation}
	where $b_{g} = \max_{\tau}\{ \|G(t_q+N-\tau)-\hat{G}(t_q+N-\tau) \|  , \|G(t_q+N-\tau)\| \} < \infty$. Notice that when $N \to \infty$,  $\xi_3 \to 0$.
	
	\noindent When $N>2N_2+1$, by combing the results from \eqref{equ:item1}, \eqref{equ:gg} and \eqref{equ:thetatheta2}, we obtain the upper bound on $  \|  {\theta}_{iq} - \hat{\theta}_{iq}\|$: 
	\begin{equation}
	\begin{aligned}
	\|  {\theta}_{iq} - \hat{\theta}_{iq}\| &= 
	\text{\circled{\scriptsize1}}+\text{\circled{\scriptsize2}}+\text{\circled{\scriptsize3}}\\
	< &  \xi_1+ \xi_2+ \xi_3\\
	=& c_{11}\bar{c}_{\pi}^N +  \frac{\bar{c}_{\pi}(1-\bar{c}_{\pi}^{N-N_2-2})}{1-\bar{c}_{\pi}} \epsilon_2 + \frac{\bar{c}_{\pi}^{N}(\bar{c}_{\pi}^{-2-N_2}-1)}{1-\bar{c}_{\pi}}b_g + c_{12}\delta(t_q+N) + c_{13}\delta(t_q+N-1) \\
	&+ \big(c_{14}\|X(t_q+N)\|\| + c_{15}\|X(t_q+N-1)\|+ \|\Pi(N)\|\big)T_{q-1}\\
	=& \zeta(N) + \psi(N)T_{q-1} ,\\
	\zeta(N)&\triangleq c_{11}\bar{c}_{\pi}^N +  \frac{\bar{c}_{\pi}(1-\bar{c}_{\pi}^{N-N_2-2})}{1-\bar{c}_{\pi}} \epsilon_2 + \frac{\bar{c}_{\pi}^{N}(\bar{c}_{\pi}^{-2-N_2}-1)}{1-\bar{c}_{\pi}}b_g + c_{12}\delta(t_q+N) + c_{13}\delta(t_q+N-1),  \\
	\psi(N)&\triangleq c_{14}\|X(t_q+N)\| + c_{15}\|X(t_q+N-1)\|+ \|\Pi(N)\|.
	\end{aligned}
	\label{equ:lemma4}
	\end{equation}
	Notice that when $N \to \infty$, $\psi(N) \to 0$ and $\zeta(N) \to 0$. 
	
	\noindent Further we consider,
	\begin{equation*}
	\begin{aligned}
	T_q-T_{q-1} &= \zeta(N) + (\psi(N)-1)T_{q-1}\\
	&=\zeta(N) + (\psi(N)-1)(\zeta(N) + \psi(N)T_{q-2})\\
	&= \psi(N)(\zeta(N) + (\psi(N)-1)T_{q-2})\\
	&= \psi(N)(T_{q-1}-T_{q-2}).\\
	\end{aligned}
	\end{equation*}
	
	\noindent We observe that when $N$ is large enough, $\psi(N)<1$, such that
	\begin{equation}
	\left| T_q-T_{q-1} \right|<\left| T_{q-1}-T_{q-2} \right|<\cdots<\left|T_1-T_0\right|.
	\end{equation}
	Notice that when $q$ is large enough, $\theta_{iq}$ converges to optimal (Lemma \ref{lemma:fullobservation}), thus, we can now draw the conclusion: for any $\xi>0$, there exist $N < \infty$ and policy improvement step $q < \infty$ such that,
	\begin{equation}
	T_q= \|  {\theta}_{iq} - \hat{\theta}_{iq}\| \leq \xi.
	\end{equation}

\end{proof}

\end{document}